\documentclass[a4paper,orivec,envcountsect,envcountsame,runningheads,final]{llncs}
\pdfoutput=1

\usepackage[utf8]{inputenc}
\usepackage[british]{babel}
\usepackage{xspace}
\usepackage{tabularx}
\usepackage{rotating}

\usepackage[inline]{enumitem}
\setlist[enumerate,1]{label=(\arabic*),font=\normalfont,align=left,leftmargin=0pt,labelindent=0pt,listparindent=\parindent,labelwidth=0pt,itemindent=!,topsep=3pt,parsep=0pt,itemsep=3pt,start=1}
\setlist[enumerate,2]{label=(\alph*),font=\normalfont,labelindent=*,leftmargin=*,start=1}
\setlist[itemize]{labelindent=*,leftmargin=*,topsep=5pt,itemsep=3pt}
\setlist[description]{labelindent=*,leftmargin=*,itemindent=-1 em}

\usepackage[%
rm={oldstyle=false,proportional=true},%
sf={oldstyle=false,proportional=true},%
tt={oldstyle=false,proportional=true,variable=true},%
qt=false%
]{cfr-lm}

\usepackage[noadjust]{cite} 
\usepackage{graphicx,cite,csquotes}

\usepackage{mathrsfs,mathtools,bm}
\numberwithin{equation}{section}
\usepackage[all,2cell]{xy}
\UseAllTwocells
\SelectTips{cm}{}
\newdir{ >}{{}*!/-5pt/@{>}}
\newdir{ (}{{}*!/-10pt/\dir^{(}}
\newdir{ )}{{}*!/-10pt/\dir_{(}}

\let\doendproof\endproof
\renewcommand\endproof{~\hfill\qed\doendproof}

\usepackage{microtype}

\spnewtheorem{defn}[theorem]{Definition}{\bfseries}{\rmfamily}
\spnewtheorem{expl}[theorem]{Example}{\bfseries}{\rmfamily}
\spnewtheorem{rem}[theorem]{Remark}{\bfseries}{\rmfamily}

\def\D{\DCat}
\def\C{\Cat}

\def\epsilon{\varepsilon}
\def\eps{\epsilon}

\renewcommand{\rho}{\varrho}
\def\ol{\overline}

\def\inj{\mathsf{in}}
\def\o{\cdot}

\def\S{\mathscr{S}}

\newcommand{\takeout}[1]{\empty}

\newcommand{\inm}{\mathsf{inm}}

\renewcommand{\phi}{\varphi}

\newcommand{\el}{\mathsf{el}\,}
\newcommand{\xra}{\xrightarrow}

\renewcommand{\ddag}{\ddagger}

\newcommand{\Set}{\mathbf{Set}}

\newcommand{\Cat}{\mathscr{C}}

\newcommand{\DCat}{\mathscr{D}}

\newcommand{\id}{\mathit{id}}
\newcommand{\Id}{\mathsf{Id}}

\newcommand{\colim}{\mathop{\mathsf{colim}}}

\newcommand{\epito}{\twoheadrightarrow}

\newcommand{\monoto}{\rightarrowtail}

\newcommand{\J}{\mathfrak{J}}

\newcommand{\Nat}{\mathbb{N}}

\newcommand{\To}{\Rightarrow}

\mathchardef\ordinarycolon\mathcode`\:
\mathcode`\:=\string"8000
\begingroup \catcode`\:=\active
  \gdef:{\mathrel{\mathop\ordinarycolon}}
\endgroup

\mathchardef\hyph="2D
\newcommand{\dash}{\mathord{-}}
  
\usepackage[utf8]{inputenc}
\usepackage[T1]{fontenc}
\usepackage{ifthen}
\newboolean{full}
\setboolean{full}{true}

\makeatletter
\newenvironment{notheorembrackets}{%
\csdef{@spopargbegintheorem}##1##2##3##4##5{\trivlist%
      \item[\hskip\labelsep{##4##1\ ##2}]{##4{##3}\@thmcounterend\ }##5}%
    }{%
\csdef{@spopargbegintheorem}##1##2##3##4##5{\trivlist%
      \item[\hskip\labelsep{##4##1\ ##2}]{##4(##3)\@thmcounterend\ }##5}%
    }
\makeatother

\usepackage{macros}

\usepackage{ifdraft}

\ifdraft{
\usepackage[marginparwidth=2cm,nohead,nofoot,
            headsep = 1cm,
            footskip = 1cm,
            text={12.2cm,19.3cm},
            paperwidth=16.2cm,
            paperheight=23.3cm,
            marginparsep=1mm,
            marginparwidth=1.8cm,
            footskip=1cm,
            centering
            ]{geometry}
}{}

\usepackage[nomargin,marginclue,footnote]{fixme}
\fxusetargetlayout{color}
\FXRegisterAuthor{sm}{asm}{SM}
\FXRegisterAuthor{ja}{aja}{JA}
\FXRegisterAuthor{hu}{ahu}{HU}

\usepackage{seqsplit}
\usepackage{xstring}
\usepackage{xcolor}

\makeatletter
\ifdraft{%
\usepackage[notcite,notref]{showkeys}%

\renewcommand*\showkeyslabelformat[1]{%
\@ifundefined{hideNextShowKeysLabel}{%
\noexpandarg%
\StrSubstitute{#1}{ }{\textvisiblespace}[\TEMP]%
\parbox[t]{\marginparwidth}{\raggedright\normalfont\small\ttfamily\(\{\){\color{red!50!black}\expandafter\seqsplit\expandafter{\TEMP}}\(\}\)}%
}{}
}

\hypersetup{hidelinks}
}{}
\makeatother

\titlerunning{On Algebras with Effectful Iteration}
\authorrunning{S.~Milius, J.~Ad\'amek, and H.~Urbat}

\title{On the Behaviour of Coalgebras with Side
  Effects and Algebras with Effectful Iteration}
\author{%
  Stefan Milius\inst{2}\fnmsep\thanks{Supported by Deutsche
    Forschungsgemeinschaft (DFG) under projects MI~717/5-2 and
    MI~717/7-1 and as part of the Research and
    Training Group 2475 ``Cybercrime and Forensic Computing'' (393541319/GRK2475/1-2019)}
  \and
  Ji\v{r}\'i Ad\'amek\inst{1}\fnmsep\thanks{Supported by the Grant
    agency of the Czech Republic under the grant 19-00902S}
  \and 
  Henning Urbat\inst{2}\fnmsep\thanks{Supported by Deutsche Forschungsgemeinschaft (DFG) under project SCHR~1118/8-2}%
}
\institute{Department of Mathematics, Faculty of Electrical
  Engineering, Czech Technical University in Prague 
  \and Lehrstuhl f\"ur Theoretische Informatik, Friedrich-Alexander-Universität Erlangen-Nürnberg}

\begin{document}
\maketitle
\begin{abstract}
For every finitary monad $T$ on sets and every endofunctor $F$ on the category of $T$-algebras we introduce the concept of an ffg-Elgot algebra for $F$, that is, an algebra admitting coherent solutions for finite systems of recursive equations with effects represented by the monad $T$. The goal is to study the existence and construction of free ffg-Elgot algebras. To this end, we investigate the locally ffg fixed point $\phi F$, i.e.~the colimit of all $F$-coalgebras with free finitely generated carrier, which is shown to be the initial ffg-Elgot algebra. This is the technical foundation for our main result: the category of ffg-Elgot algebras is monadic over the category of $T$-algebras.
\end{abstract}

\section{Introduction}

Terminal coalgebras yield a fully abstract domain of behavior for a
given kind of state-based systems whose transition type is described
by an endofunctor $F$. Often one is mainly interested in the study of the
semantics of \emph{finite} coalgebras. For instance, regular languages are the
behaviors of finite deterministic automata, while the terminal
coalgebra of the corresponding functor is formed by \emph{all} formal
languages. For endofunctors on sets, the \emph{rational fixed point}
introduced by Ad\'amek, Milius and Velebil~\cite{amv_atwork} yields a
fully abstract domain of behavior for finite coalgebras. However, in
recent years there has been a lot of interest in studying coalgebras
over more general categories than sets. In particular, categories of
algebras for a (finitary) monad $T$ on sets are a paradigmatic
setting; they are used, for instance, in the generalized
determinization framework of Silva et al.~\cite{sbbr13} and yield
\emph{coalgebraic language equivalence}~\cite{bms13} as a semantic
equivalence of coalgebraic systems with side effects modelled by the monad
$T$. In the category $\C$ of $T$-algebras, several notions of 'finite'
object are natural to consider, and each yields an ensuing
notion of 'finite' coalgebra: (1)~free objects on finitely many generators
(\emph{ffg} objects) yield precisely the coalgebras that are the target of
generalized determinization; (2)~finitely presentable (\emph{fp}) objects are the
ones that can be presented by finitely many generators and relations
and yield the rational fixed point; and (3)~finitely
generated (\emph{fg}) objects, which are the ones presented by finitely many
generators (but possibly infinitely many relations). Taking the
colimits of all coalgebras with ffg, fp, and fg carriers,
respectively, yields three coalgebras $\phi F$, $\rho F$ and
$\theta F$ which, under suitable assumptions on $F$, are all fixed
points of~$F$~\cite{amv_atwork,mpw16,Urbat17}. Our present paper
is devoted to studying the fixed point $\phi F$, which we call the \emph{\name
fixed point} of $F$. For a finitary endofunctor $F$ preserving
surjective and non-empty injective morphisms in $\C$, the three fixed points
are related to each other and the terminal coalgebra $\nu F$ as follows:
\begin{equation}\label{eq:pic}
  \phi F \epito \rho F \epito \theta F \monoto \nu F,
\end{equation}
where $\epito$ denotes a quotient coalgebra and $\monoto$ a subcoalgebra. The
three right-hand fixed points are characterized by a universal property both
as a coalgebra and (when inverting their coalgebra structure) as an
algebra~\cite{amv_atwork,m_cia,mpw16}; see~\cite{Urbat17} for one
uniform proof. We recall this in more detail in
Section~\ref{sec:four}.

The main contribution of the present paper is a new characterization
of the \name fixed point $\phi F$ by a universal property as an
algebra. As already observed by Urbat~\cite{Urbat17}, as a coalgebra,
$\phi F$ does not satisfy the expected finality property since
coalgebra homomorphisms from coalgebras with ffg carrier into $\phi F$
may fail to be unique. A simple initiality property of $\phi F$ as an
algebra was recently established by
Milius~\cite[Theorem~4.4]{milius18}: $\phi F$ is the initial ffg-Bloom
algebra for $F$, where an \emph{ffg-Bloom algebra} is an $F$-algebra
equipped with an operation that assigns to every $F$-coalgebra carried
by an ffg object a coalgebra-to-algebra morphism subject to a
functoriality property. Equivalently, the ffg-Bloom algebras for $F$
form the slice category
$\phi F/\alg F$~\cite[Proposition~4.5]{milius18}. Here we introduce
the notion of an \emph{ffg-Elgot algebra}
(Section~\ref{sec:ffgElgot}), which is an algebra for $F$ equipped
with an operation that allows to take solutions of \emph{effectful
  iterative equations} (see \autoref{rem:eff}) subject to two natural
axioms. These axioms are inspired by and closely related to the axioms
of (ordinary) Elgot algebras~\cite{amv_elgot}, which we recall in
Section~\ref{sec:elgot}. We then prove that $\phi F$ is the initial
ffg-Elgot algebra (\autoref{thm:ini}), which strengthens the previous
initiality result.

In addition, we study the construction of \emph{free} ffg-Elgot
algebras. In the case of ordinary Elgot algebras, it was
shown~\cite{amv_elgot} that the rational fixed point $\rho(F(-)+Y)$ is
a free Elgot algebra on $Y$. In addition, the category of Elgot
algebras is the Eilenberg-Moore category for the corresponding monad
on $\C$.  In the present paper, we prove that free ffg-Elgot
algebras exist on every object $Y$ of $\C$. But is it true that the
free ffg-Elgot algebra on $Y$ is $\phi(F(-)+Y)$? We do not know the
answer for arbitrary objects $Y$, but if $Y$ is a free $T$-algebra (on
a possibly infinite set of generators), the answer is affirmative
(\autoref{thm:free}).

Finally, we prove that the category of ffg-Elgot algebras is monadic
over $\C$, i.e.~ffg-Elgot algebras are precisely the Eilenberg-Moore
algebras for the monad that assigns to a given object $Y$ of $\C$ its
free ffg-Elgot algebra (\autoref{thm:mon}).

This paper is a revised and extended version of our conference
paper~\cite{amu18_cmcs} containing full proofs.

\paragraph{Related Work and History.} While our new notion of an
ffg-Elgot algebra is directly based on the previous notion of Elgot
algebra~\cite{amv_elgot}, studying operators taking solutions of
recursive equation systems and their properties goes back a long
way. The most well-known examples of such structures are probably the
iteration theories of Bloom and \'Esik~\cite{BE93} whose work is based
on Elgot's seminal work~\cite{Elgot75} on the semantics of iterative
specifications. Algebras for iteration were first studied by
Nelson~\cite{nelson} (see also Tiuryn~\cite{tiuryn80} for a related
concept). Our work grows out of the coalgebraic approach to the
semantics of iteration which started with Moss' work~\cite{Moss01} on
parametric corecursion. Independently, and almost at the same time, it
was also realized by Ghani et al.~\cite{glmp_cmcs, glmp} and Aczel et
al.~\cite{aav,aamv} that final coalgebras for parametrized functors
$F(-) + Y$ give rise to a monad, whose structure generalizes
substitution of infinite trees over a signature. Later it was shown by
Milius~\cite{m_cia} that one can approach this monad through algebras
with unique solutions of recursive equations. The monad arising from
the parametrized rational fixed points $\rho(F(-) + Y)$ was introduced
in~\cite{amv_atwork} based on a category-theoretic generalization of
Nelson's notion of iterative algebra. This generalizes Courcelle's
regular trees~\cite{courcelle} and their substitution. The monad of
free ffg-Elgot algebras is a new example of a monad arising from
parametrized coalgebras.

\paragraph{Outline of the Paper.} We begin in \autoref{sec:prelim} by
recalling a number of preliminaries, e.g.~on varieties and `finite'
objects in such categories. This material might be skipped by readers
who are familiar with it. We also recall background on the four fixed
points in~\eqref{eq:pic}, and, as a first highlight, we present in
\autoref{P:uses} an example of the locally ffg fixed point $\phi F$ in
a setting where the other three are trivial.

\autoref{sec:elgot} is a brief recap on Elgot algebras and so can be
skipped by expert readers who have seen them before.

The concept of ffg-Elgot algebras is introduced in
\autoref{sec:ffgElgot}. Readers who would like to see the connection
of ffg-Elgot algebras to effectful iterative equations should jump
right to \autoref{rem:eff}, where this connection is explained. The
main technical results of our paper then follow as already
explained. First, \autoref{thm:ini} shows that $\phi F$ is the initial
ffg-Elgot algebra. Second, \autoref{thm:initialfy} establishes, for a
free object $Y$ of our base variety $\C$, a one-to-one correspondence
of pairs consisting of an ffg-Elgot algebra $A$ for $F$ and a morphism
$Y \to A$ with ffg-Elgot algebras for $F(-)+ Y$. This result turns out
to be a key ingredient of the construction of free ffg-Elgot algebras
from coalgebras for $F(-) + Y$ (see \autoref{C:free} and
\autoref{thm:free}) for a free object $Y$. Monadicity of
ffg-Elgot algebras is etablished in \autoref{S:mon}.

We conclude the paper in \autoref{S:con}.

Finally, in the short appendix a technical result concerning the
construction of $\phi F$ is presented.

\paragraph{Acknowedgements.} We would like to thank the anonymous
reviewers whose suggestions helped us to improve our paper.

\section{Preliminaries}
\label{sec:prelim}

\subsection{Varieties and `Finite' Algebras}
\label{sec:vars}

Throughout the paper we will work with a (finitary, many-sorted)
variety $\C$ of algebras and an endofunctor $F$ on
it. Equivalently,~$\C$ is the category of Eilenberg-Moore algebras for
a finitary monad $T$ on the category $\Set^S$ of $S$-sorted
sets~\cite{arv}. We will speak about objects of $\C$ (rather than
algebras for $T$) and reserve the word 'algebra' for algebras for
$F$. All the `usual' categories of algebraic structures and their
homomorphisms are varieties: monoids, (semi-)groups, rings, vector
spaces over a fixed field, modules for a (semi-)ring, positive convex
algebras, join-semilattices, Boolean algebras, distributive lattices,
and many others. In each case, the corresponding monad $T$ assigns to
a set the free object on it, e.g.~$TX = X^*$ for monoids, the finite
power-set monad $T= \powf$ for join-semilattices, and the
subdistribution monad $\D$ for positive convex algebras, etc.

\takeout{%
We are interested in the study of the semantics of 'finite' iterative
systems of equations. But several notions of finiteness present
themselves as natural in a variety. Finite algebras, of course;
however, using them as variable objects for iteration does not yield
anything new w.r.t.~working with finite sets. Besides, three other
notions are of interest: finitely generated, finitely presentable and
free finitely generated algebras. All three notions are classical
notions from general algebra, which have a categorical
characterization that we now recall:}

As mentioned in the introduction, every variety $\C$ of algebras comes
with three natural notions of 'finite' objects, each of which admits a
neat category-theoretic characterization (see \cite{arv}):

\paragraph{Finitely presentable objects} (fp objects, for short) can
be presented by finitely many generators and relations. An object $X$
is fp iff the covariant hom-functor $\C(X,-)\colon \C \to Set$ is
\emph{finitary}, i.e.~it preserves filtered colimits. Recall that a
category $\D$ is filtered if every finite subcategory has a cocone in
$\D$, and a diagram is filtered if its scheme is a filtered category.\smnote{I vote
  for repeating the definition of filtered; not all our readers know
  this.} 
We denote by
$\C_\fp$ the full subcategory of $\C$ given by all fp
objects.\smnote{I vote for including the next sentence; we need this
  in our proofs!} In our proofs we will use the well-known fact that
every object $X$ is the filtered colimit of the canonical diagram
$\C_\fp/X \to \C$, i.e.~objects in the diagram scheme are
morphisms $P \to X$ in $\C$ with $P$ fp.

\paragraph{Finitely generated objects} (fg objects, for short) are
presented by finitely many generators but, possibly, infinitely many
relations. An object $X$ is fg iff $\C(X,-)$ preserves filtered
colimits with monic connecting morphisms. Hence, every fp object is fg
but not conversely. In fact, the fg objects are precisely the
(regular) quotients of the fp objects~\cite[Proposition~5.22]{arv}. 

\paragraph{Free finitely generated objects} (ffg objects, for short)
are the objects $(TX_0, \mu_{X_0})$ where $X_0$ is a finite $S$-sorted
set (i.e.~the coproduct of all components $X_s$, $s \in S$ is
finite). An object $X$ is a split quotient of an ffg object iff
$\C(X,-)$ preserves \emph{sifted}
colimits~\cite[Corollary~5.14]{arv}.\smnote{I vote for 
  including precise pointers to make life easy for reader wanting to
  locate proofs.} Recall from~\cite{arv} that sifted
colimits are more general than filtered colimits: a sifted colimit is
a colimit of a diagram $D\colon \D \to \C$ whose diagram scheme $\D$ is a
sifted category, which means that finite products commute with
colimits over $\D$ in $\Set$.  \smnote{I vote for including the
  precise definition.} \iffull More precisely, $\D$ is sifted iff given any
diagram $D\colon \D \times \J \to \Set$, where $\J$ is a finite discrete
category, the canonical map
\[
  \colim\limits_{d \in \D} \Big(\prod\limits_{j \in \J}
  D(d,j)\Big) 
  \to
  \prod\limits_{j \in \J} (\colim\limits_{d \in \D} D(d,j))
\]
is an isomorphism. 
\fi%
For instance, every filtered category and every category with finite
coproducts is sifted~\cite[Example~2.16]{arv}.

The category $\C$ is cocomplete and the forgetful functor $\C \to \Set^S$
preserves and reflects sifted colimits, that is, sifted colimits in $\C$ are formed on the level of underlying sets \cite[Proposition~2.5]{arv}.

\begin{rem}\label{R:sifted}
  A finitely cocomplete category has sifted colimits if and only if it
  has filtered colimits and reflexive coequalizers, i.e.~coequalizers
  of parallel pairs of epimorphisms with a joint splitting. Moreover a
  functor preserves sifted colimits if and only if it preserves filtered
  colimits and reflexive coequalizers~\cite{AdamekEA10}.
\end{rem}

We denote by $\C_\ffg$ the full subcategory of ffg objects of
$\C$. Analogously to the fact that every object of $\C$ is a filtered
colimit of fp objects, every object $X$ is a sifted colimit of the
canonical diagram $\C_\ffg/X \to \C$; this follows
from~\cite[Proposition~5.17]{arv}.

\subsection{Relation between the object classes.} 
\smnote{Even though not needed for our proofs, I still vote for
  inluding this short subsec as it gives background for repeating the
  relation of the four fixed points; which we should absolutely
  recall!}  We already mentioned that every fp object is fg (but not
conversely, in general). Clearly, every ffg object is fp, but not
conversely in general (e.g.~consider any fp monoid which is not of the
form $X^*$ for some finite set $X$). So, in general, we have full
embeddings
\[
  \C_\ffg \overset{\neq}{\subto} \C_\fp \overset{\neq}{\subto} \C_\fg.
\]
In rare cases, all three object classes coincide; e.g.~in $\Set$
(considered as a variety) and the category of vector spaces over a field.

The equation $\C_\fg = \C_\fp$ holds true, for example, for all
locally finite varieties (i.e.~where ffg objects are carried by finite
sets), e.g.~Boolean algebras, distributive lattices or
join-semilattices); for positively convex algebras~\cite{SokolovaW15},
commutative monoids \cite{redei,freyd68}, abelian groups, and more
generally, in any category of (semi-)modules for a semiring $\S$ that
is \emph{Noetherian} in the sense of \'Esik and
Maletti~\cite{em_2011}. That means that every subsemimodule of an fg semimodule
is fg itself. For example, the following semirings are Noetherian:
every finite semiring, every field, every principal ideal domain such
as the ring of integers and therefore every finitely generated
commutative ring by Hilbert's Basis Theorem. The tropical semiring
$(\Nat \cup \{\infty\}, \min, +, \infty, 0)$ is \emph{not}
Noetherian~\cite{em_2010}. The usual semiring of natural numbers is
not Noetherian either, but for the category of $\Nat$-semimodules ($=$
commutative monoids), $\C_\fp = \C_\fg$ still holds.

\subsection{Functors and Liftings}
\smnote{This subsec is important background info which we should have
  to show that our results carry some relevance!}

We will consider coalgebras for functors $F$ on the variety $\C$. In
many cases $F$ is a \emph{lifting} of a functor on many-sorted sets, i.e.~there is a
functor $F_0\colon \Set^S \to \Set^S$ such that the square below
commutes, where $U\colon \C \to \Set^S$ denotes the forgetful functor.
\[
\xymatrix{
  \C \ar[r]^-F \ar[d]_U
  & \C \ar[d]^U \\
  \Set^S \ar[r]_-{F_0} & \Set^S
}
\]
It is well-known~\cite{applegate,johnstone_lift} that liftings of a
given functor $F_0$ on $\Set^S$ to $\C$, the variety given by the
monad $(T,\eta,\mu)$, are in bijective correspondence with
distributive laws of that monad over the functor $F_0$.
This means natural transformations $\lambda\colon TF_0 \to F_0T$
such that the following two diagrams commute:
\[
  \xymatrix{
    F_0
    \ar[d]_{\eta F_0}
    \ar[rd]^-{F_0\eta}
    \\
    TF_0 \ar[r]_-{\lambda}
    &
    F_0T
  }
  \qquad
  \xymatrix{
    TTF_0
    \ar[d]_{\mu F_0}
    \ar[r]^-{T\lambda}
    &
    TF_0T \ar[r]^-{\lambda T}
    &
    F_0TT
    \ar[d]^{F_0\mu}
    \\
    TF_0
    \ar[rr]_-{\lambda}
    &&
    F_0T
    }
\]
Given a distributive law $\lambda$ of $T$ over $F_0$, the corresponding
lifting $F$ assigns to a $T$-algebra $(A,a)$ the $T$-algebra
$(F_0A, F_0a \cdot \lambda_A)$. 
It was observed by Turi and Plotkin~\cite{TuriP97} that a final
coalgebra for $F_0$ lifts to a final coalgebra for the lifting $F$%
\iffull\else\/, and this is then the final \emph{bialgebra} for the
corresponding distributive law\fi. 
\iffull%
Indeed, denoting by
$\xi\colon \nu F_0 \to F_0(\nu F_0)$ the final coalgebra for $F_0$, we obtain a
canonical $T$-algebra structure on $\nu F_0$ by corecursion, i.e.~as the unique
coalgebra homomorphism $a\colon T(\nu F_0) \to \nu F_0$ in the diagram below:
\[
  \xymatrix{
    T(\nu F_0) 
    \ar[r]^-{T\xi} 
    \ar[d]_a
    &
    TF_0(\nu F_0)
    \ar[r]^-{\lambda_{\nu F_0}}
    &
    F_0T(\nu F_0) 
    \ar[d]^{F_0a}
    \\
    \nu F_0
    \ar[rr]_-{\xi}
    &&
    F_0(\nu F_0)
  }
\]
It is easy to verify that $a$ is an Eilenberg-Moore algebra and that
this turns $\nu F_0$ into the final coalgebra for the lifting $F$. Note
that the above square expresses that $(\nu F_0, a, \xi)$ is a
\emph{$\lambda$-bialgebra}, and it is the final one~\cite{TuriP97}.
\fi

Coalgebras for lifted functors are significant because the
targets of \emph{finite} coalgebras $X$ under \emph{generalized
  determinization}~\cite{sbbr13} are precisely those coalgebras for
the lifting $F$ that are carried by ffg objects $(TX,\mu_X)$. In more detail,
generalized determinization is the process of turning a given
coalgebra $c\colon X \to F_0TX$ in $\Set^S$ into a coalgebra for the
lifting $F$: one uses the freeness of $TX$ and the fact that $FTX$
is a $T$-algebra to extend $c$ to a $T$-algebra homomorphism
$c^*\colon TX \to FTX$. The \emph{coalgebraic language
  semantics}~\cite{bms13} of $(X,c)$ is then the final semantics of
$c^*$. A classical instance of this is the language semantics of
non-deterministic automata considered as coalgebras
$X \to \{0,1\} \times (\powf X)^\Sigma$; here the generalized
determinization with $T = \powf$ and $F_0 = \{0,1\} \times X^\Sigma$ on
$\Set$ is the well-known subset construction turning a
non-deterministic automaton into a deterministic one.

\subsection{Four Fixed Points}
\label{sec:four}

Fixed points of a functor $F$ are (co)algebras whose structure is
invertible. Let us now consider a finitary endofunctor
$F\colon \C\to \C$ on our variety.  Then $F$ has a terminal coalgebra
\cite[Theorem 6.10]{amm18}, which we denote by $\nu F$. Its coalgebra
structure $\nu F \to F(\nu F)$ is an isomorphism by Lambek's
lemma~\cite{lambek}, and so $\nu F$ is a fixed point of $F$. The
terminal coalgebra $\nu F$ is fully abstract w.r.t.~behavioural
equivalence: given $F$-coalgebras $(X,c)$ and $(Y,d)$, two states
$x \in X$ and $y\in Y$ are called \emph{behavioural equivalent} if
there exists a pair of coalgebra homomorphisms
$f\colon (X, c) \to (Z,e)$ and $g\colon (Y,d) \to (Z,e)$ such that
$f(x) = g(y)$. Behavioural equivalence instantiates to well-known
notions of indistinguishability of system states, e.g.~for
$F = \powf$, it is strong bisimilarity of states in finitely branching
transitions systems, and for $FX = \{0,1\} \times X^\Sigma$ it yields
the language equivalence of states in deterministic automata. One can
show that two states are behaviourally equivalent if and only if they
are identified under the unique coalgebra homomorphisms into $\nu F$.

There are three further fixed points of $F$ obtained from `finite'
coalgebras, where `finite' can mean each of the three notions
discussed in Subsection~\ref{sec:vars}. More precisely, denote by
\[
  \coa F
\]
the category of all $F$-coalgebras. We consider its full subcategories
given by all coalgebras with fp, fg, and ffg carriers, respectively,
and we denote them as shown below:
\[
  \coafr F \subto \coaf F \subto \coafg F \subto \coa F. 
\] 
Since the three subcategories above are
essentially small, we can form coalgebras as the colimits of
the above inclusions as follows:
\begin{align*}
  \phi F & = \colim (\coafr F \subto \coa F),\\
  \theta F & = \colim (\coafg F \subto \coa F), \\
  \rho F & = \colim (\coaf F \subto \coa F).
\end{align*}
Note that the latter two colimits are filtered; in fact, $\coafg F$
and $\coaf F$ are clearly closed under finite colimits in $\coa F$,
whence they are filtered categories. The first colimit is a sifted
colimit since its diagram scheme $\coafr F$ is closed under finite
coproducts~\cite[Lemma~3.6]{milius18}. In what follows, the
objects of $\coafr F$ are called \emph{ffg-coalgebras}. 

We now discuss the three coalgebras above in more detail.

\paragraph{The rational fixed point} is the coalgebra $\rho F$. This
is a fixed point as proved by Ad\'amek, Milius and
Velebil~\cite{amv_atwork}. In addition, $\rho F$ is characterized by a
universal property both as a coalgebra and as an algebra:
\begin{enumerate}
\item As a coalgebra, $\rho F$ is the terminal \emph{locally finitely
    presentable} (lfp) coalgebra, where a coalgebra is called lfp if
  it is a filtered colimit of a diagram formed by coalgebras from
  $\coaf F$~\cite{m_linexp}.
\item As an algebra, $\rho F$ is the initial iterative algebra for
  $F$.
\end{enumerate}
An \emph{iterative algebra} is an $F$-algebra $a\colon FA \to A$ such
that every \emph{fp-equation}, i.e.~a morphism $e\colon X \to FX + A$
with $X$ fp, has a unique \emph{solution} in $A$. The latter means
that there exists a unique morphism $\sol e$ such that the following
square commutes\footnote{Note that in a diagram we usually denote
  identity morphisms simply by the (co)domain object.}:
\begin{equation}\label{eq:sol}
  \vcenter{
    \xymatrix@C+1pc{
      X 
      \ar[r]^-{\sol e}
      \ar[d]_e
      &
      A
      \\
      FX + A \ar[r]_-{F\sol e + A}
      &
      FA + A \ar[u]_{[a,A]}
    }
  }
\end{equation}
This notion is a categorical
generalization of iterative $\Sigma$-algebras for a single-sorted signature $\Sigma$ originally introduced by Nelson~\cite{nelson};
see also Tiuryn~\cite{tiuryn80} for a closely related concept.

\paragraph{The locally finite fixed point} is the coalgebra
$\theta F$. This coalgebra was recently introduced and studied by
Milius, Pattinson and Wi\ss\/mann~\cite{mpw16,mpw20} for a finitary
endofunctor $F$ preserving non-empty monos. They proved $\theta F$ to
be a fixed point of $F$ and characterized by two universal properties
analogous to the rational fixed point:
\begin{enumerate}
\item As a coalgebra, $\theta F$ is the terminal \emph{locally
    finitely generated} (lfg) coalgebra, where a coalgebra is called
  lfg if it is a colimit of a directed diagram of coalgebras in
  $\coafg F$.
\item As an algebra, $\theta F$ is the initial fg-iterative algebra
  for $F$, where fg-iterative is simply the variation of iterative
  above where the domain object of $e\colon X \to FX + A$ is required
  to be fg in lieu of fp.
\end{enumerate}
Moreover, $\theta F$ is always a subcoalgebra of
$\nu F$~\cite[Theorem~3.10]{mpw20} and thus fully abstract
w.r.t.~behavioral equivalence.

\paragraph{The \name fixed point} is the coalgebra $\phi F$. Recently,
Urbat~\cite{Urbat17} has proved that $\phi F$ is indeed a fixed point
of $F$, provided that $F$ preserves sifted colimits. Actually, in
\emph{loc.~cit.} the coalgebra $\phi F$ is defined to be the colimit
of all $F$-coalgebras whose carrier is a split quotient of an ffg
object. However, this is the same colimit as above, as we prove in the
Appendix.\smnote{I vote for not including this appendix here as this
  is a rather technical discussion that would blow up this paragraph.}

Moreover, \emph{loc.~cit.~}provides a general framework that allows to
prove uniformly that all four coalgebras $\rho F$, $\rho F$,
$\theta F$ and $\nu F$ are fixed points. In addition, a
uniform proof of the universal properties of $\rho F$, $\theta F$ and
$\nu F$ is given.

Somewhat surprisingly, the coalgebra $\phi F$ fails to have the finality property
w.r.t.~to coalgebras in $\coafr F$:
Urbat~\cite[Example~4.12]{Urbat17}
gives such a counterexample, see \autoref{S:phif} below.
This also shows that $\phi F$ cannot have a universal property as some
kind of iterative algebra (i.e.~where solutions are unique). 

\paragraph{Relations between the Fixed Points.} Recall that a
\emph{quotient} of a coalgebra is represented by a coalgebra
homomorphism carried by a regular epimorphism ($=$ surjective algebra
morphism) in $\C$. Suppose we have a finitary functor $F$ on $\C$
preserving surjective morphisms and non-empty injective
ones.\footnote{These are mild assumptions; e.g.~if $\C$ is
  single-sorted and $F$ a lifting of a set functor, then these
  conditions are fulfilled.} Then the subcoalgebra $\theta F$ of
$\nu F$ is a quotient of $\rho F$, which in turn is a quotient of
$\phi F$~\cite{mpw20,milius18}; see~\eqref{eq:pic}:
\[
  \phi F \epito \rho F \epito \theta F \monoto \nu F.
\]
Whenever $\C_\fp = \C_\fg$, we clearly have $\coaf F = \coafg F$ and
hence $\rho F \cong \theta F$ (i.e.~$\rho F$ is fully abstract
w.r.t.~behavioral equivalence). If $\C_\fp = \C_\fg = \C_\ffg$,
$\rho F$ and $\theta F$ coincide with $\phi F$ as well. Moreover,
Milius~\cite{milius18} introduced the notion of a \emph{proper}
functor (generalizing the notion of a proper semiring of \'Esik and
Maletti \cite{em_2010}) and proved that a functor $F$ is proper if and
only if the three fixed points coincide, i.e.~the
picture above collapses to
$\phi F \cong \rho F \cong \theta F \subto \nu F$. Loc.~cit.~also
shows that on a variety $\C$ where fg objects are closed under taking
kernel pairs, every endofunctor mapping kernel pairs to weak pullbacks
in $\Set$ is proper~\cite[Proposition~5.10]{milius18}.\footnote{Note
  that these conditions are fulfilled in particular by every locally
  finite variety and every category of semirings for a Noetherian
  semiring and any lifted endofunctor whose underlying $\Set$ functor
  preserves weak pullbacks.}
\smnote{This does make sense! It means that $UF$ maps kernel pairs to weak pullbacks.}

\paragraph{Instances of the three fixed points}
$\phi F$, $\theta F$ and $\rho F$ have mostly been considered for
proper functors (where the three are the same, e.g.~for functors on
$\Set$), or else on algebraic categories where $\C_\fp = \C_\fg$
(where $\rho F \cong \theta F$, i.e.~the rational and locally finite
fixed points coincide). We shall see in \autoref{S:phif} that $\phi F$
can be different from $\rho F$ and $\theta F$ (even when the latter
two are isomorphic). Before that we illustrate the relationship
of $\rho F$ and $\theta F$ to $\nu F$ by a number of well-known
important examples:  \iffull
\smnote{I strongly vote for mentioning the following examples!}
\begin{examples}
\begin{enumerate}
\item For the set functor $FX = \{0,1\} \times X^\Sigma$, whose
  coalgebras are deterministic automata with the input
  alphabet $\Sigma$, the terminal coalgebra is formed by all formal
  languages on $\Sigma$ and the three fixed points are formed by
  all regular languages.
\item For a signature $\Sigma = (\Sigma_n)_{n< \omega}$ of operation
  symbols with prescribed arity we have the associated polynomial
  endofunctor on $\Set$ given by
  $F_\Sigma X = \coprod_{n < \omega} \Sigma_n \times X^n$. Its terminal
  coalgebra is carried by the set of all (finite and infinite)
  $\Sigma$-trees, i.e.~rooted and ordered trees where each node with
  $n$-children is labelled by an $n$-ary operation symbol. The three
  fixed points are all equiv to the subcoalgebra given by \emph{rational} (or
  regular~\cite{courcelle}) $\Sigma$-trees, i.e.~those $\Sigma$-trees
  that have only finitely many different subtrees (up to isomorphism)
  This characterization is due to Ginali~\cite{ginali}. For
  example, for the signature $\Sigma$ formed by a binary operation symbol
  $*$ and a constant $c$ the following infinite $\Sigma$-tree (here
  written as an infinite term) is rational:
  \[
    c * (c * (c* \cdots )));
  \]
  in fact, up to isomorphism its only subtrees are the whole tree and the single-node
  tree labelled by $c$).
\item Consider the endofunctor $FX= \S \times X^\Sigma$ on the category of semimodules for
  the semiring $\S$. The fixed point $\theta F$, which is isomorphic to $\rho F$ if $\S$ is
  Noetherian, is formed by all formal power series (i.e.~elements of
  $\S^{\Sigma^*}$) recognizable by finite $\S$-weighted
  automata. From the  Kleene-Schützenberger
  theorem~\cite{schuetzenberger} (see also~\cite{br88}) it follows
  that these are, equivalently, the \emph{rational} formal
  power-series.
\item For $FX = k \times X$ on $\Set$ the terminal coalgebra is carried
  by the set $k^\omega$ of all streams on $k$, and the three fixed
  points are equal; they are formed by all eventually periodic streams (also called
  lassos). If $k$ is a field, and we consider $F$ as a functor on
  vector spaces over $k$, we obtain rational streams~\cite{rutten_rat}.
\item Recall~\cite{Doberkat06} that a \emph{positively
    convex algebra} is a set $X$ equipped with finite convex sum
  operations. This means that for every $n$ and $p_1,
  \ldots, p_n \in [0,1]$ with $\sum_{i = 1}^n p_i \leq 1$ we
  have an $n$-ary operation assigning to $x_1, \ldots, x_n \in X$
  an element $\bigboxplus\limits_{i=1}^n p_i x_i$ subject to the following
  axioms:
  \begin{enumerate}
  \item $\bigboxplus\limits_{i=1}^n p_i^kx_i = x_k$ whenever $p_k^k = 1$ and $p_i^k = 0$ for $i \neq k$, and
  \item $\bigboxplus\limits_{i=1}^n p_i \left(\bigboxplus\limits_{j=1}^k q_{i,j} x_j\right) = \bigboxplus\limits_{j=1}^k \left(\sum\limits_{i = 1}^n p_iq_{i,j}\right) x_j$.
  \end{enumerate}
  For $n = 1$ we write the convex sum operation for $p \in [0,1]$
  simply as $px$. Positively convex algebras together with maps
  preserving convex sums in the obvious sense form the category
  $\PCA$. Note that $\PCA$ is (isomorphic to) the
  Eilenberg-Moore category for the monad $\mathcal D$ of finitely
  supported subprobability distributions on sets.

  Sokolova and Woracek~\cite{SokolovaW18} have recently proved that
  the functor $FX = [0,1] \times X^\Sigma$ and its subfunctor $\hat F$
  mapping a set $X$ to the set of all pairs $(o,f)$ in $[0,1] \times
  X^\Sigma$ satisfying
  \[
    \forall s \in \Sigma: \exists p_s\in [0,1], x_s \in X:
    o + \sum\limits_{s\in \Sigma} p_s \leq 1, f(s) = p_sx_s
  \]
  are proper functors on $\PCA$. Hence, for those functors our three
  fixed points coincide. In particular, the latter functor $\hat F$ is
  used to capture the complete trace semantics of generative
  probabilistic transition systems~\cite{SilvaS11}. Hence, for $\hat
  F$, our three fixed points collect precisely the probabilistic
  traces of finite such systems. 

\item Given an alphabet $\Sigma$, for the functor
  $FX = \{0,1\} \times X^\Sigma$ on the category of idempotent
  semirings the locally finite fixed point $\theta F$ is formed by all
  context-free languages~\cite{mpw20}. Descriptions of $\rho F$ and
  $\phi F$ are unknown in this case.

  More generally, consider first the category of associative
  $\S$-algebras for the commutative semiring $\S$,
  i.e.~$\S$-semimodules equipped with an additional monoid structure
  such that multiplication is an $\S$-semimodule morphism in each of
  its arguments. This is the Eilenberg-Moore category for the monad
  $\Poly{-}$ assigning to each set $X$ the set of $\S$-polynomials of
  over $X$, i.e.~functions $X^* \to \S$ with finite support. This is not
  quite the category $\C$, but one considers $\Sigma$-pointed
  $\S$-algebras, where $\Sigma$ is an input alphabet,
  i.e.~$\S$-algebras $A$ equipped with a map $\Sigma \to A$. The
  corresponding monad is $\Poly{-+\Sigma}$. The terminal coalgebra for
  the functor $FX = \S \times X^\Sigma$ on $\C$ is again carried by
  the set of all formal power series over $\Sigma$, and the locally
  finite fixed point $\theta F$ is formed by all constructively
  $\S$-algebraic formal power-series~\cite{mpw16}. (The original definition of those
  power-series goes back to Fliess~\cite{Fliess1974}, see
  also~\cite{DrosteEA09}; an equivalent coalgebraic characterization
  was first provided by Winter et al.~\cite{jcssContextFree}.)
\end{enumerate}
\end{examples}
\else \smnote{There was a longer version of the following here; you
  can switch it on with the \textsf{full} switch.} For example,
regular languages for the automaton functor $2 \times (-)^\Sigma$ on
$\Set$; rational formal power series for the functor
$\S \times (-)^\Sigma$ on $\S$-semimodules (whenever $\S$ is a proper
semiring the three fixed points coincide); rational (a.k.a.~regular)
$\Sigma$-trees for the polynomial functor on $\Set$ associated to the
signature $\Sigma$; eventually periodic and rational streams for the
functor $k \times (-)$ on $\Set$ and vector spaces over the field $k$,
respectively; the behaviors of probabilistic automata modelled as
coalgebras for $[0,1] \times (-)^\Sigma$ on the category of positive
convex algebras (that this functor is proper was recently proved by
Sokolova and Woracek~\cite{SokolovaW18}); finally, (deterministic)
context-free languages and constructively $\S$-algebraic formal
power-series (the weighted counterpart of context-free
languages)~\cite{mpw16}. Note that the last two examples are instances of
the locally finite fixed point $\theta F$, but a description of
$\phi F$ and $\rho F$ is unkown.  \fi

\begin{rem}
  The rational fixed point $\rho F$ and the locally finite one,
  $\theta F$, are defined and studied
  more generally than in the present setting, namely for finitary
  functors $F$ on a locally finitely presentable category $\C$ (see
  Ad\'amek and Rosick\'y~\cite{ar} for an introduction to locally
  presentable categories); see~\cite{amv_atwork,m_linexp} for $\rho F$
  and~\cite{mpw16,mpw20} for $\theta F$.

  The following are instances of $\rho F$ and
  $\theta F$ for $F$ on a locally finitely presentable category $\C$:
  \iffull
  \begin{enumerate}
  \item Consider the functor category $\Set^{\mathcal F}$, where
    $\mathcal F$ is the category of finite sets and maps and denote by
    $V\colon \mathcal F \subto \Set$ is the full embedding. Further,
    consider the endofunctor $FX = V + X \times X + \delta(X)$ with
    $\delta(X)(n) = X(n+1)$. This is a paradigmatic example
    of a functor arising from a \emph{binding signature} for which
    initial semantics was studied by Fiore et al.~\cite{fpt}.

    The final coalgebra $\nu F$ is carried by the presheaf of all
    $\lambda$-trees modulo
    $\alpha$-equivalence~\cite{amv_horps_full}. In fact, the functor
    $\nu F$ assigns to $n$ the set of all (finite and infinite)
    $\lambda$-trees in $n$ free variables (note that such a tree may
    have infinitely many bound variables). Moreover, $\rho F$ is carried by
    the rational $\lambda$-trees, where an $\alpha$-equivalence class
    is called \emph{rational} if it contains at least one
    $\lambda$-tree which has (up to isomorphism) only finitely many
    different subtrees (see~\emph{op.~cit.}).

    The coalgebra of all $\lambda$-trees with finitely many free
    variables modulo $\alpha$-equivalence also appears as the final
    coalgebra for a very similar functor on the category of nominal
    sets~\cite{KurzEA13}. Moreover, the rational $\lambda$-trees form
    its rational fixed point~\cite{mw15}. Similarly for any functor
    on nominal sets arising from a binding
    signature~\cite{KurzEA13,msw16}.
    
  \item Courcelle's algebraic trees~\cite{courcelle} occur as a
    locally finite fixed point. In more detail, fix a polynomial
    functor $H_\Sigma\colon\Set \to \Set$ and consider the category
    $\C = H_\Sigma/\Mndf(\Set)$ of $H_\Sigma$-pointed finitary monads
    $M$ on $\Set$, i.e.~those equipped with a natural transformation
    $H_\Sigma \to M$. The assignment $M \mapsto H_\Sigma M + \Id$
    provides an endofunctor $F\colon \C \to \C$ whose terminal
    coalgebra is carried by the monad $T_\Sigma$ assigning to a set
    $X$ the set of all $\Sigma$-trees over $X$. The locally finite
    fixed point $\theta F$ is the monad $A_\Sigma$ of algebraic
    $\Sigma$-trees~\cite{mpw16}. Note that in this category $\C$, fp
    and fg objects do not coincide. Hence, it is unclear whether
    $\theta F$ and $\rho F$ are isomorphic.
  \end{enumerate}
  \else \smnote{Again there was a longer version of the following
    here; you can switch it on with the \textsf{full} switch.}
  (a)~Courcelle's algebraic trees~\cite{courcelle} as proved
  in~\cite{mpw16}; (b) rational $\lambda$-trees (modulo
  $\alpha$-equivalence) for a functor on the category of presheaves
  over finite sets~\cite{amv_horps_full} or for a related functor on
  the category of nominal sets~\cite{mw15}; more generally,
  (c)~rational trees over an arbitrary binding signature (see Fiore et
  al.~\cite{fpt}) as proved in~\cite{msw16}. Again, (a) is an instance
  of the locally finite fixed point $\theta F$ but a description of
  the rational fixed point is unknown. \fi %
  In the setting of general locally finitely presentable categories, there is
  no analogy to $\phi F$, of course.
\end{rem}

\subsection{A Nontrivial Example of the Locally ffg Fixed Point}
\label{S:phif}

We now present a new example where only $\phi F$ is interesting
whereas the other three fixed points are trivial.

We consider the monad $T$ on $\Set$ whose algebras are the algebras
with one unary operation $u$ (with no equation):
\[ 
  TX = \Nat\times X \quad\text{with}\quad u(n,x)=(n+1,x). 
\]
The unit $\eta$ and multiplication $\mu$ of this monad are
given by $\eta_X(x) = (0,x)$ and $\mu_X(n, (m,x)) = (n+m, x)$. Since
$TX$ is the free algebra with one unary operation on $X$, its elements
$(n,x)$ correspond to terms $u^n(x)$. Let $F$ be the identity
functor $\Id$ on the category $\C = \Set^T$. The final coalgebra for
$\Id$ is lifted from $\Set$: it is the trivial algebra on $1$ with
$\id_1$ as its coalgebra structure. Since $1$ is clearly finitely
presented by one generator $x$ and the relation $u(x) = x$, both of
the diagrams $\coaf \Id$ and $\coafg \Id$ have a terminal object. This
is then their colimit, whence $\rho \Id \cong \theta \Id \cong 1$.
  
However, $\phi \Id$ is non-trivial and interesting. An ffg-coalgebra
$TX\xra{\gamma} TX$ may be viewed (by restricting it to its generators
in $X$) as obtained by generalized determinization of an
$FT$-coalgebra with $F = \Id$ on $\Set$, i.e.~a map
$X\xra{\langle o,\delta\rangle} \Nat\times X$ that we call
\emph{stream coalgebra}.  Given a state $x\in X$, we call the sequence
of natural numbers
\[
  (\,o(x), o(\delta(x)), o(\delta^2(x)),\ldots\,)
\]
the \emph{stream generated by $x$}. Since the set $X$ is finite, this
stream is eventually periodic, i.e. of the form $s=s_0s_1^\omega$ for
finite lists $s_0$ and $s_1$ of natural numbers. (Here
$(\dash)^\omega$ means infinite iteration.) Two eventually periodic
streams $s=s_0s_1^\omega$ and $t=t_0t_1^\omega$ with
$s_1=(s_{1,0},\ldots, s_{1,p-1})$ and
$t_1 = (t_{1,0},\ldots,t_{1,q-1})$ are called \emph{equivalent} if one
has
\begin{equation}\label{eq:equivstreams} 
  q\o \sum_{i<p} s_{1,i} = p\o \sum_{j<q} t_{1,j},
\end{equation} 
i.e.~the two lists $s_1$ and $t_1$ have the same arithmetic mean (or,
equivalently, the entries of the two lists $s_1^q$ and $t_1^p$ of
length $p \cdot q$ have the same sum). For instance, the streams
\[ s=(1,2,7,4)(1,3,2)^\omega = (1,2,7,4,1,3,2,1,3,2,1,3,2,\ldots)\]
and
\[ t=(5,6)(0,4)^\omega = (5,6,0,4,0,4,0,4,0,4,\ldots) \] are
equivalent. Note that the above notion of equivalence is well-defined,
i.e.~not depending on the choice of the finite lists $s_0,s_1$ and
$t_0,t_1$ in the representation of $s$ and $t$. In fact, given
alternative representations $s=\ol s_0 \ol s_1^\omega$ and
$t=\ol t_0 \ol t_1^\omega$ with
$\ol s_1 = (\ol s_{1,0},\ldots,\ol s_{\ol p-1})$ and
$\ol t_1=(\ol t_{1,0},\ldots,\ol t_{1,\ol q-1})$, the lists
$s_1^{\ol p}$ and $\ol s_1^p$ are equal up to cyclic shift, as are the
lists $t_1^{\ol q}$ and $\ol t_1^{q}$. Therefore from
\eqref{eq:equivstreams} it follows that
\[
  \ol q\o q\o p \o \sum_{i<\ol p} \ol s_{1,i}
  =
  \ol q\o q \o \ol p\o \sum_{i<p} s_{1,i}
  =
  \ol q \o \ol p \o p\o \sum_{j<q}t_{1,j}
  =
  \ol p \o p\o q\o \sum_{j<\ol q} \ol t_{1,j}.
\]
Dividing by $p\o q$ yields the required result:
\[
  \ol q\o \sum_{i<\ol p} \ol s_{1,i}
  =
  \ol p\o \sum_{j<\ol q} \ol t_{1,j}.
\]

\begin{rem}\label{R:new}
  \begin{enumerate}
  \item\label{R:new:1} In the proof of \autoref{P:uses} further below
    we use the following well-known fact about colimits of sets.
    For every diagram $D\colon \D \to \Set$, a cocone
    $c_i\colon Di \to C$ ($i \in \D$) is a colimit iff~(a) the
    colimit injections $c_i$ are jointly surjective,
    i.e.~$C = \bigcup c_i[Di]$, and (b)~given $c_i(x) = c_j(y)$ for
    some pair $x \in Di, y \in Dj$, there exists a zig-zag of morphisms
    of $\D$ whose $D$-image connects $x$ and $y$.
  \item\label{R:new:2} Moreover, if $D$ is a filtered diagram, then condition~(b) can
    be substituted by the condition that when two elements 
    $x,y \in Di$ are merged by $c_i$ then they are also
    merged by $Dh\colon Di \to Dj$ for some morphism $h\colon i \to j$ of $\D$.
  \end{enumerate}
\end{rem}

\begin{proposition}\label{P:uses}
  The coalgebra $\phi \Id$ is carried by the set of equivalence
  classes (cf.~\eqref{eq:equivstreams}) of eventually periodic streams. 
\end{proposition}
In more detail, the unary operation and
the coalgebra structure are both given by
$\id\colon \phi \Id \to \phi \Id$, and for every $\Id$-coalgebra
$(TX,\gamma_X)$ with $X$ finite, the colimit injection
$\gamma_X^\sharp\colon TX\to \phi\Id$ maps $(m,x)\in TX$ to the
equivalence class of the stream generated by $x$.
\proof
  \begin{enumerate}
  \item We first show that the above morphisms $(\dash)^\sharp$ form a
    cocone. Given an ffg-coalgebra $(TX,\gamma_X)$ for $\Id$ and elements
    $(m,x),(n,y)\in TX$ with $\gamma_X(m,x)=(n,y)$, the stream
    generated by $y$ is the tail of the stream generated by $x$, and
    thus the two streams are equivalent. This shows that $\gamma_X^\sharp$
    is a coalgebra homomorphism.

    To show that the morphisms $(\dash)^\sharp$ form a cocone,
    suppose that $h\colon (TX,\gamma_X)\to (TY,\gamma_Y)$ is a
    homomorphism in $\coafr \Id$, and let $(m,x)\in TX$ and
    $(n,y)\in TY$ with $h(m,x)=(n,y)$ be given. We need to show that
    the streams generated by $x$ and $y$ are equivalent. Denote by
    \begin{equation}\label{eq:reachablestates}
      (m_j,x_j) := \gamma_X^j(m,x)
      \quad\text{and}\quad
      (n_j,y_j):= \gamma_Y^j(n,y)\qquad (j=0,1,2,\ldots)
    \end{equation}
    the states reached from $(m,x)$ and $(n,y)$, resp., after $j$ steps.
    Since $h$ is a coalgebra homomorphism, one has
    $h(m_j,x_j)=(n_j,y_j)$ for all $j$. Since $X$ is finite, there
    exist natural numbers $k\geq 0$ and $p>0$ with $x_k=x_{k+p}$. Then
    the eventually periodic stream generated by $x$ is given by
    \[
      (m_1-m_0,m_2-m_1,\ldots, m_k-m_{k-1})(m_{k+1}-m_k,\ldots,
      m_{k+p}-m_{k+p-1})^\omega
    \] 
    Since $h(m_k,x_k)=(n_k,y_k)$ and
    $h(m_{k+p},x_{k+p}) = (n_{k+p},y_{k+p})$, one has $y_k=y_{k+p}$,
    which implies that $y$ generates the stream
    \[
      (n_1-n_0,n_2-n_1,\ldots, n_k-n_{k-1})(n_{k+1}-n_k,\ldots,
      n_{k+p}-n_{k+p-1})^\omega
    \]
    To show that the streams generated by $x$ and $y$ are equivalent,
    it suffices to verify that $m_{k+p}-m_k=n_{k+p}-n_k$, as this
    entails that
    \[
      p\o \sum_{i<p} m_{k+i+1}-m_{k+i} = p\o (m_{k+p}-m_k) = p\o
      (n_{k+p}-n_k) = p\o \sum_{i<p} n_{k+i+1}-n_{k+i}.
    \]
    To prove the desired equation, we compute
    \begin{align*}
      (n_{k+p},y_{k+p}) &= h(m_{k+p},x_{k+p}) \\
      &= h(m_{k+p},x_k)  \\
      &= h(m_{k+p}-m_k+m_k, x_k)  \\
      &= (m_{k+p}-m_k+n_k,y_k) 
    \end{align*}
    where the last equality uses that $h(m_k,x_k)=(n_k,y_k)$ and that
    $h$ is a morphism of $\C$. This implies $n_{k+p}=m_{k+p}-m_k+n_k$.
    
\item We prove that the cocone $(\dash)^\sharp$ is a colimit cocone. Since
  sifted colimits in $\coa\Id$ are formed as in $\C$ and thus as in
  $\Set$, we can apply \autoref{R:new}: we will show that (a) the morphisms $\gamma_X^\sharp$ are
  jointly surjective and (b) given ffg-coalgebras $(TX,\gamma_X)$ and
  $(TY,\gamma_Y)$ and two states $(m,x)\in TX$ and $(n,y)\in TY$
  merged by $\gamma_X^\sharp$ and $\gamma_Y^\sharp$, there exists a zig-zag in
  $\coafr\Id$ connecting the two states. Statement (a) is clear
  because finite stream coalgebras generate precisely the eventually
  periodic streams. For (b), we adapt the argument of the first part
  of our proof and continue to use the notation
  \eqref{eq:reachablestates}. Since $X$ and $Y$ are finite, there
  exist natural numbers $k\geq 0$ and $p>0$ with $x_k=x_{k+p}$ and
  $y_k=y_{k+p}$. As the streams generated by $x$ and $y$ are
  equivalent, one has $m_{k+p}-m_k = n_{k+p}-n_k$. Consider the ffg-coalgebra
  $(TZ, \gamma_Z)$ with $Z=\{z_0,z_1,\ldots,z_{k+p-1}\}$,
  and $\gamma_Z$ defined on the generators by
  \[ 
    \gamma_Z(z_j) = (0,z_{j+1}) \;\;(j<k+p-1) 
    \quad\text{and}\quad
    \gamma_Z(z_{k+p-1}) = (m_{k+p}-m_k,z_k). 
  \] 
  Form the morphisms $g\colon TZ\to TX$ and $h\colon TZ\to TX$ given
  on generators by
  \[ 
    g(z_{j}) = (m_j,x_j) 
    \quad\text{and}\quad 
    h(z_j) = (n_j,y_j)\qquad (j<k+p).
  \] 
  Then $g$ is a coalgebra homomorphism. Indeed, for $j<k+p-1$
  we have
  \begin{align*}
    g(\gamma_Z(z_j)) &= g(0,z_{j+1}) & \text{(def. $\gamma_Z$)}\\
    &= (m_{j+1},x_{j+1}) & \text{(def. $g$)}\\
    &= \gamma_X(m_j,x_j) & \text{(def. $m_{j+1}$,\, $x_{j+1}$)}\\
    &= \gamma_X(g(z_{j})) & \text{(def. $g$)}
  \end{align*}
  and moreover
  \begin{align*}
    g(\gamma_Z(z_{k+p-1})) 
    &= g(m_{k+p}-m_k,z_k) & \text{(def. $\gamma_Z$)}\\
    &= (m_{k+p}-m_k+m_k,x_k) & \text{(def. $g$)}\\
    &= (m_{k+p},x_{k+p})\\
    &= \gamma_X(m_{k+p-1},x_{k+p-1}) & \text{(def. $m_{k+p},\,x_{k+1}$)}\\
    &= \gamma_X(g(z_{k+p-1})). & \text{(def. $g$)}
  \end{align*}
  Analogously for $h$. Thus we have constructed a zig-zag
  \[
    (TX,\gamma_X) \xleftarrow{g} (TZ,\gamma_Z) \xrightarrow{h} (TY,\gamma_Y)
  \]
  in $\coafr \Id$ connecting $(m,x)$ and $(n,y)$, as required.\qed
\end{enumerate}
\doendproof
Observe that every non-empty ffg-coalgebra
$(TX,\gamma_X)$ admits infinitely many coalgebra homomorphisms into
$\phi \Id$. For instance, any constant map into $\phi \Id$ is
one. This shows that, in general, the coalgebra $\phi F$ is not final
w.r.t.~the ffg-coalgebras.

\section{Recap: Elgot Algebras}
\label{sec:elgot}

In this section we briefly recall the notion of an Elgot
algebra~\cite{amv_elgot} and some key results in order to contrast this
with our subsequent development of ffg-Elgot algebras in
Section~\ref{sec:ffgElgot}. Throughout this section we assume the
endofunctor $F\colon \C \to \C$ to be finitary.

Recall from \autoref{sec:four} that an \emph{fp-equation} is a morphism
\[
  e\colon X \to FX + A,
\]
where $X$ is an fp object (of variables) and $A$ an arbitrary
object of \emph{parameters}.

Furthermore, if $A$ carries the structure of an $F$-algebra $a\colon FA \to
A$, then a \emph{solution} of $e$ in $A$ is a morphism $\sol e\colon X \to
A$  such that the square~\eqref{eq:sol} commutes.
\begin{notation}
  We use the following notation for fp-equations: 
  \begin{enumerate}
  \item Given an fp-equation $e\colon X \to FX + A$ and a morphism $h\colon A
    \to B$ we have an fp-equation
    \[
      h \aft e = (\,X \xrightarrow{e} FX + A \xrightarrow{FX + h} FX +
      B\,).
    \]
  \item Given a pair of fp-equations $e\colon X \to FX + Y$ and
    $f\colon Y \to FY + Z$ we combine them into the following fp-equation
    \[
      e \sq f = (\,X+Y \xrightarrow{[e,\inr]} FX +Y \xrightarrow{FX + f}
      FX + FY + Z \xrightarrow{\can + Z} F(X+Y)+Z\,),
    \]
    where $\can = [F\inl,F\inr]\colon FX + FY \to F(X+Y)$ denotes the
    canonical morphism.
  \end{enumerate}
\end{notation}
\begin{notheorembrackets}
\begin{defn}[{\cite{amv_elgot}}]\label{def:props}
  An \emph{Elgot algebra} is a triple $(A,a,\dagger)$ where $(A,a)$ is
  an $F$-algebra and $\dagger$ is an
  operation
  \[
    \frac{e\colon X \to FX + A}{\sol e\colon X \to A}
  \]
  assigning to every fp-equation in $A$ a solution, subject to the
  following two conditions:
  \begin{enumerate}
  \item \emph{Weak Functoriality.} Given a pair of fp-equations
    $e\colon X \to FX + Z$ and $f\colon Y \to FY + Z$, where $Z$ is an fp object,
    and a coalgebra homomorphism $m\colon X \to Y$ for $F(-) + Z$, then for
    every morphism $h\colon Z \to A$ we have
    $\sol{(h \aft f)} \o m = \sol{(h \aft e)}$:
    \[
      \vcenter{
        \xymatrix{
          X \ar[r]^-e \ar[d]_m & FX + Z \ar[d]^{Fm + Z} \\
          Y \ar[r]_-f & FY + Z
        }
      }
      \quad
      \implies
      \quad
      \vcenter{
        \xymatrix@R-1.5pc{
          X \ar[rd]^-{\sol{(h \aft e)}} \ar[dd]_m
          \\
          & A.
          \\
          Y \ar[ru]_-{\sol{(h \aft f)}}
        }
      }
    \]
  \item \emph{Compositionality.} For every pair of fp-equations
    $e\colon X \to FX + Y$ and $f\colon Y \to FY + A$ we have
    \[
      \sol{(\sol f \aft e)} = (X \xrightarrow{\inl} X + Y
      \xrightarrow{\sol{(e\sq f)}} A).
    \]
  \end{enumerate}
\end{defn}
\end{notheorembrackets}
\begin{rem}\label{rem:square_bullet_properties}
  Later we will need the following properties of $\aft$ and
  $\sq\,$: 
  \begin{enumerate}
  \item $t \aft (s \aft e) = (t \o s) \aft e$ for every $e\colon X \to FX +
    A$, $s\colon A \to B$ and $t\colon B \to C$; 
  \item $s \aft (e \sq f) = e \sq\, (s\aft f)$ for every $e\colon X \to FX +
    Y$, $f\colon Y\to FY + A$ and $s\colon A \to B$; 
  \item $(e \sq f) \sq g = (\inl \aft e) \sq (f \sq g)$ for every $e\colon X
    \to FX + Y$, $f\colon Y \to FY + Z$ and $g\colon Z \to FZ + V$. 
  \end{enumerate}
  For the proof of the first two see~\cite[Remark~4.6]{amv_elgot}. The
  remaining one is easy to prove by considering the three coproduct
  components of $X+Y+Z$ separately. We leave this as an exercise for the
  reader.
\end{rem}

Note that, in lieu of weak functoriality, $\dagger$ was previously
required to satisfy (full) functoriality~\cite{amv_elgot};  this
states that for every pair of 
fp-equations $e\colon X \to FX + A$, $f\colon Y \to FY + A$ and a
coalgebra homomorphism $m\colon (X,e) \to (Y,f)$ we have
$\sol f \o m = \sol e\colon X \to A$. However, this makes no
difference:
\begin{lemma}
  Functoriality and Weak Functoriality are equivalent properties of
  $\dagger$. 
\end{lemma}
\proof Functoriality clearly implies Weak Functoriality. In order to
prove the converse, let $e\colon X \to FX + A$, $f\colon Y \to FY + A$
be fp-equations, and let $m\colon (X,e) \to (Y,f)$ be a coalgebra
morphism. Given an algebra $(A,a)$, write $A$ as the filtered colimit
of its canonical diagram $\C_\fp /A$ (cf.~Section~\ref{sec:vars}). The
functor $FX + (-)$ preserves filtered colimits, and so $FX + A$ is the
filtered colimit of the diagram formed by all morphisms
$FX + h\colon FX + Z \to FX + A$, where $h$ ranges over
$\C_\fp/A$. Since $X$ is fp, the morphism $e\colon X \to FX + A$ factors
through one of these morphisms, i.e.~there exists a morphism
$h\colon Z \to A$ with $Z$ fp and $e'\colon X \to FX + Z$ such that
$e = h \aft e'$:
  \[
    \xymatrix{
      X \ar[r]^-e \ar[rd]_-{e'} & FX + A \\
      & FX + Z \ar[u]_{FX + h}
    }
  \]
  Similarly, we have a factorization of $f\colon Y \to FY + A$, and by
  filteredness of the diagram $\C_\fp/A \to \C$, we can assume that the
  same $h\colon Z \to A$ is used. Thus a morphism
  $f'\colon Y \to FY + Z$ is given such that
  $h \aft f' = (FY + h) \o f' = f$. We do not claim that $m$ is a
  coalgebra homomorphism from $(X,e')$ to $(Y,f')$. However, the
  corresponding equation holds when postcomposed by the colimit
  injection $FY + h$:
  \begin{align*}
    (FX + h)\o (Fm + Z) \o e' 
    & = (Fm + A) \o (FX + h) \o e' \\
    & = (Fm + A) \o e \\
    & = f \o m \\
    & = (FY + h) \o f' \o m.
  \end{align*}
  By \autoref{R:new}\ref{R:new:2}, there exists a morphism $h'\colon Z' \to A$ with $Z'$ fp and a
  connecting morphism $z\colon Z \to Z'$ in $\C_\fp/A$, i.e.~$z$ satisfies
  $h' \o z = h$, such that $FY + z$ merges $(Fm + Z) \o e'$ and $f' \o
  m$. It follows that $m$ is a coalgebra homomorphism from $z \aft e'$ to
  $z \aft f'$. Indeed, in the following diagram
  \[
    \xymatrix@C+1pc{
      X \ar[r]^-{e'} \ar[d]_m
      & 
      FX + Z 
      \ar[r]^-{FX + z} 
      \ar[d]^{Fm + Z}
      & 
      FX + Z'
      \ar[d]^{Fm + Z'}
      \ar@{<-} `u[l] `[ll]_{z \aft e'} [ll]
      \\
      Y \ar[r]_-{f'}
      &
      FY + Z
      \ar[r]_-{FY + z}
      &
      FY + Z'
      \ar@{<-} `d[l] `[ll]^-{z \aft f'} [ll]
    }
  \]
  the left-hand square commutes when postcomposed with $FY +
  z$; thus, since the upper and lower parts as well as the right-hand square
  commute, so does the outside, as desired. By Weak Functoriality, we
  thus conclude
  \begin{align*}
    \sol f \o m 
    & = \sol{(h \aft f')} \o m = \sol{((h'\o z) \aft f')}\o m 
      = \sol{(h' \aft (z \aft f'))} \o m \\
    & = \sol{(h' \aft (z\aft e'))} = \sol{((h'\o z)\aft e')} 
      = \sol{(h \aft e')} = \sol e. \tag*{\qed}
  \end{align*}
\doendproof

\begin{examples}
  Let us recall a few examples of Elgot algebras~\cite{amv_elgot}.
  \begin{enumerate}
  \item Iterative $F$-algebras (cf.~Section~\ref{sec:four}):
    the operation $\dagger$ assigning to every equation its unique
    solution satisfies Compositionality and (Weak)
    Functoriality, see \cite[2.15--2.19]{amv_elgot}. It follows that $\rho F$,
    $\theta F$ and $\nu F$ are Elgot algebras. 

  \item Cpo enrichable algebras. Recall that a \emph{complete partial
      order} (\emph{cpo}, for short) is a partially ordered set having
    joins of $\omega$-chains. Cpos form a category $\CPO$ whose
    morphisms are the \emph{continuous} functions, i.e.~functions
    preserving joins of $\omega$-chains. Let $F_0\colon \Set \to \Set$
    be a functor having a \emph{locally continuous} lifting
    $F\colon \CPO \to \CPO$, i.e.~a lifting such that the derived
    mappings $\CPO(X,Y) \to \CPO(FX,FY)$ are continuous for all cpos
    $X$ and $Y$. (For example, every polynomial functor $F_\Sigma$
    associated to the signature $\Sigma$ has a lifting to $\CPO$.)
    
    Suppose further that $a\colon FA \to A$ is an algebra where $A$ is
    a cpo with a least element $\bot$ and $a$ is continuous. Then $A$
    is an Elgot algebra w.r.t.~the operation $\dagger$ assigning to
    an fp-equation its least solution. More precisely, given an fp-equation
    $e\colon X \to F_0X + A$ (in $\Set$), consider $X$ as a cpo with
    discrete order. Then we obtain the following continuous endomap on
    $\CPO(X,A)$, the cpo of continuous functions from $X$ to $A$:
    \[
      h \mapsto [a,A] \o (Fh + A) \o e
    \]
    (cf.~\eqref{eq:sol}), and we let $\sol e$ be its least fixed point
    (which exists by Kleene's fixed point theorem). For details
    see~\cite[3.5--3.8]{amv_elgot}. 
  \item CMS enrichable algebras. A related example is based on
      \emph{complete metric spaces}, i.e.~metric
      spaces in which every Cauchy sequence has a limit. Here one
      considers the category $\CMS$ of complete metric spaces with distances
      in $[0,1]$ and non-expanding maps, i.e.~maps $f\colon X \to Y$ such
      that for every $x,x' \in X$ one has $d_Y(fx,fx') \leq
      d_X(x,x')$. Note that for two complete metric spaces $X$ and $Y$
      the set of non-expanding maps $\CMS(X,Y)$ forms a complete metric
      space with the supremum metric
      \[
        d_{X,Y}(f,g) = \sup\limits_{x \in X} d_Y(f(x),g(x)).
      \]
      Let $F_0\colon \Set \to \Set$ be a functor having a \emph{locally
        contracting} lifting to $\CMS$, i.e.~a lifting $F\colon \CMS \to
      \CMS$ for which there exists some $\eps < 1$ such that for
      all $f, g\colon X \to Y$ in $\CMS$ one has
      \[
        d_{X,Y} (f,g) \leq \eps d_{FX,FY}(Ff,Fg). 
      \]
      (Again, polynomial set functors have locally contracting liftings to
      $\CMS$.)
      
      Now suppose that $a\colon FA \to A$ is a non-empty algebra such that $A$
      carries a complete metric space and $a$ is a non-expanding
      map. Then $A$ is iterative, whence an Elgot algebra. In fact,
      for every equation $e\colon X \to FX + A$ consider $X$ as a discrete
      metric space (i.e.~all distances are $1$) and consider the 
      endofunction on $\CMS(X,A)$ given by
      \[
        h \mapsto [a,A] \o (Fh + A) \o e,
      \]
      which is $\eps$-contracting for the $\eps$ above. Then, by
      Banach's fixed point theorem, this function has a unique fixed
      point, viz.~the unique solution of $e$. For details
      see~\cite[2.8--2.11]{amv_elgot}.
    \item As a concrete instance of the previous point one can obtain
      fractals as solutions of equations. For example, let $A$ be the
      set of closed subsets of the unit interval $[0,1]$ equipped with
      the following binary operation:
      \[
        (C,C') \mapsto \frac{1}{3}C \cup \left(\frac{1}{3}{C'} + \frac{2}{3}\right),
      \]
      where $\frac{1}{3}C = \{\frac{1}{3} c \mid c \in C\}$ etc. Then
      $A$ is an algebra for $F_0X = X \times X$ on $\Set$, and this $F_0$
      has the locally contracting lifting
      $F(X,d) = (X \times X, \frac{1}{3}d_{\max})$, where $d_{\max}$
      denotes the usual maximum metric on the cartesian product. One
      sees that $A$ is an algebra for $F$ when equipped with the
      so-called Hausdorff metric. Hence, it is an Elgot algebra. For
      example, let $X = \{x\}$ and let $e\colon X \to FX + A$ be given by
      $e(x) = (x,x)$. Then $\sol e(x)$ is the well-known Cantor set.
  \end{enumerate}
\end{examples}
We have already mentioned in Section~\ref{sec:four} that the rational fixed
point $\rho F$ is an initial iterative $F$-algebra. Moreover, for
every object $Y$, the rational fixed point $\rho(F(-) + Y)$ is a free
iterative algebra on $Y$. Thus, the object assignment
$Y \mapsto \rho(F(-) + Y)$ yields a monad $R$ on $\C$. 
\begin{theorem}[\cite{amv_elgot}]
  The category of Eilenberg-Moore algebras for the monad $R$ is isomorphic to
  the category of Elgot algebras for $F$. 
\end{theorem}
Thus, in particular, $\rho(F(-) + Y)$ is not only a free iterative
algebra, but it is also a free Elgot algebra on $Y$, whence $\rho F$
is the initial Elgot algebra. 

\section{FFG-Elgot Algebras}
\label{sec:ffgElgot}

The rest of our paper is devoted to studying the fixed point $\phi F$,
the colimit of all ffg-coalgebras for $F$, in its own right and
establish a universal property of it as an algebra. Recall that by a
variety $\C$ we mean a finitary, many sorted variety. That is, $\C$
is (isomorphic to) the category of Eilenberg-Moore algebras for a
finitary monad $T$ on $\Set^S$, where $S$ is a set of sorts. 

\begin{assumption}
  Throughout the rest of the paper we assume that $\C$ is a variety of
  algebras and that $F\colon \C \to \C$ is an endofunctor preserving sifted
  colimits. 
\end{assumption}
\takeout{
This includes, for the monad $T$ representing $\C$, all functors that
are liftings of finitary set functor $F_0$ (i.e.~with a distributive
law of $T$ over $F_0$). Indeed, finitary set functors $F_0$ preserve
all sifted colimits~\cite[Proposition~6.30]{arv}. Since the forgetful
functor $U\colon \C \to \Set$ preserves and reflects sifted colimits,
it follows that every lifting of $F_0$ preserves sifted colimits,
too. Other examples of endofunctors on $\C$ preserving sifted colimits
are: $FX = X + X$, more, generally, any coproduct of functors with the
property, and if $\C$ is \emph{entropic} (i.e.~symmetric monoidal
closed), $FX = X \otimes X$, more generally, any finite tensor product
of functors with this property. Therefore, all (tensor-)polynomial
functors on $\C$ preserve sifted colimits.
}
\begin{examples}
  \begin{enumerate}
  \item For the monad $T$ representing $\C$, all functors that are
    liftings of a finitary functor $F_0$ on $\Set^S$ (via a
    distributive law of $T$ over $F_0$) preserve sifted
    colimits. Indeed, finitary functors $F_0: \Set^S \to \Set^S$
    preserve them~\cite[Proposition~6.30]{arv}. Since the forgetful
    functor $U\colon \C \to \Set^S$ preserves and reflects sifted
    colimits, it follows that every lifting of $F_0$ preserves sifted
    colimits, too.

    The following examples are not liftings of set
    functors.
  \item The functor $FX = X + X$, where $+$ denotes the coproduct of
    $\C$, preserves sifted colimits. More generally, every coproduct
    of sifted-colimit preserving functors preserves them
    too. Similarly for finite products of sifted-colimit preserving
    functors. Thus, all polynomial functors on $\C$ preserve sifted colimits.  
  \item Let $\C$ be an \emph{entropic} variety (see
    e.g.~\cite{DaveyD85}) aka~\emph{commutative} variety~(see
    e.g.~\cite{Linton66}), i.e.~such that the usual tensor product $\otimes$
    (representing bimorphisms) makes it a symmetric monoidal
    closed category. (Examples include sets, vector spaces,
    join-semi\-lattices, or abelian groups.) Then the functor
    $FX = X \otimes X$ preserves sifted colimits. To see this, it
    suffices to show that (a)~$F$ is finitary and (b)~it preserves
    reflexive coequalizers (see \autoref{R:sifted}).  First note that
    since $\C$ is symmetric monoidal closed, we know that each functor
    $X \otimes -$ and $- \otimes X$ is a left adjoint and therefore
    preserves all colimits.

    Ad~(a). Suppose that $D: \D \to \C$ is a
    filtered diagram with colimit injections $a_d: Dd \to A$ for
    $d \in \D$. We need to prove that all
    $a_d \otimes a_d: Dd \otimes Dd \to A \otimes A$ form a colimit
    cocone. That is, for every morphism
    $f: X \to A \otimes A$ with $X$ fp, (i)~there exists some $d \in \D$ and
    $g: X \to Dd \otimes Dd$ with $(a_d \otimes a_d) \cdot g = f$ and
    (ii)~given $g, h: X \to Dd \otimes Dd$ that yield $f$ in this
    way, there exists a morphism $m: d \to d'$ in $\D$ such that
    $Dm \otimes Dm$ merges $g$ and $h$~\cite[Lemma~2.6]{amsw19_1}. 

    To prove (i), we use that
    $- \otimes A$ is finitary to obtain some $d \in \D$ and
    $f': X \to A \otimes Dd$ with $(A \otimes a_d) \cdot f' = f$. Now
    use that $Dd \otimes -$ is finitary to obtain $d' \in \D$ and
    $f'': X \to Dd \otimes Dd'$ with
    $(Dd \otimes a_{d'}) \cdot f'' = f'$. Since $\D$ is filtered, we
    can choose morphisms $m: d \to \bar d$ and $n: d' \to \bar d$
    in $\D$. Let $g = (Dm \otimes Dn)\cdot f''$. Then we have
    \begin{align*}
      (a_{\bar d} \otimes a_{\bar d}) \cdot g
      &=
      (a_{\bar d} \otimes a_{\bar d}) \cdot (Dm \otimes Dn)\cdot f''
      =
        (a_{d} \otimes a_{d'}) \cdot f''\\
      &= (a_d \otimes A) \cdot (Dd \otimes a_{d'}) \cdot f''
      = (a_d \otimes A) \cdot f' = f
    \end{align*}
    as desired. 

    For (ii), use first that $-
    \otimes A$ is finitary and choose some morphism $o: d \to d'$ such
    that
    \[
      (Do\otimes A) \cdot \left((Dd \otimes a_d) \cdot g\right)
      =
      (Do\otimes A) \cdot \left((Dd \otimes a_d) \cdot h\right).
    \]
    It follows that $(Dd' \otimes a_d)$ merges
    $(Do \otimes Dd) \cdot g$ and $(Do \otimes Dd) \cdot h$. Now use
    that $Dd' \otimes -$ is finitary and choose a morphism
    $p:d \to d''$ in $\D$ such that $(Dd' \otimes Dp)$ also merges
    those two morphisms. Finally, use that $\D$ is filtered to choose
    two morphisms $q: d' \to \bar d$ and $r: d'' \to \bar d$ such that
    $q \cdot o = r \cdot p$, and let us call this last morphism $m: d \to
    \bar d$. Then $Dm \otimes Dm$ merges $g$
    and $h$:
    \begin{align*}
      (Dm \otimes Dm) \cdot g
      &= (D(q \cdot o) \otimes D(r \cdot p))\cdot g
        = (Dq \otimes Dr) \cdot (Do \otimes Dp) \cdot g
      \\
      &=
      (Dq \otimes Dr) \cdot (Dd' \otimes Dp) \cdot (Do \otimes Dd)
        \cdot g
      \\
      &= (Dq \otimes Dr) \cdot (Dd' \otimes Dp) \cdot (Do \otimes Dd)
        \cdot h
      \\
      &= (Dm \otimes Dm) \cdot h.
    \end{align*}
    
    Ad~(b). Let $f, g: A \to B$ be a (not necessarily reflexive) pair, and let
    $c: B \to C$ be its coequalizer. Use that all functors
    $- \otimes X$ and $X \otimes -$ preserve coequalizers to see that
    in the following diagram, whose parts commute in the obvious way,
    all rows and columns are coequalizers:
    \[  
      \xymatrix@+1pc{
        A \otimes A
        \ar@<3pt>[r]^-{f \otimes A}\ar@<-3pt>[r]_-{g \otimes A}
        \ar@<3pt>[d]^{A \otimes f}\ar@<-3pt>[d]_{A\otimes g}
        &
        B \otimes A
        \ar[r]^-{c \otimes A}
        \ar@<3pt>[d]^{B \otimes f}\ar@<-3pt>[d]_{B\otimes g}
        &
        C \otimes A
        \ar@<3pt>[d]^{C \otimes f}\ar@<-3pt>[d]_{C\otimes g}
        \\
        A \otimes B
        \ar@<3pt>[r]^-{f \otimes B}\ar@<-3pt>[r]_-{g \otimes B}
        \ar[d]^{A \otimes c}
        &
        B \otimes B
        \ar[r]^-{c \otimes B}
        \ar[d]^{B \otimes c}
        &
        C \otimes B
        \ar[d]^-{C \otimes c}
        \\
        A \otimes C
        \ar@<3pt>[r]^-{f \otimes C}\ar@<-3pt>[r]_-{g \otimes C}
        &
        B \otimes C
        \ar[r]^-{c \otimes C}
        &
        C \otimes C
      }
    \]
    By the `3-by-3 lemma'~\cite[Lemma~0.17]{Johnstone77}, it follows
    that the diagonal yields a coequalizer too, i.e.~$c \otimes c$ is
    a coequalizer of the pair $f \otimes f, g \otimes g$, as desired.
    
  \item Combining the previous argument with induction, we see that
    sifted-colimit preserving functors on an entropic variety $\C$ are
    stable under finite tensor products. Thus, all tensor-polynomial
    functors on $\C$ preserve sifted colimits.
  \end{enumerate}
\end{examples}

Under our assumptions we know that $\phi F$ is a fixed point of
$F$~\cite{Urbat17}, and we will henceforth denote the inverse of its
coalgebra structure by $t\colon F(\phi F) \to \phi F$. The following
is a variation of \autoref{def:props} where the variable objects $X$ are now
restricted to be ffg objects:
\begin{defn}\label{def:eqsol}
  By an \emph{ffg-equation} is meant a morphism $e\colon X \to FX + A$
  where $X$ is an ffg object (of \emph{variables}) and $A$ an arbitrary
    object (of \emph{parameters}). An \emph{ffg-Elgot algebra} is a
    triple $(A,a,\dagger)$ where $(A,a)$ is an $F$-algebra and
    $\dagger$ is an operation
  \[
    \frac{e\colon X \to FX + A}{\sol e\colon X \to A}
  \]
  assigning to every ffg-equation in $A$ a solution (cf.~\eqref{eq:sol}) and satisfying
  Weak Functoriality~\ref{def:props}(1) and
  Compositionality~\ref{def:props}(2) with $X, Y$ and $Z$ restricted
  to ffg objects.
\end{defn}
\begin{rem}\label{R:app}
  \begin{enumerate}
  \item Note that in categories where fp objects are ffg, e.g.~in the
    category of sets or vector spaces, (ordinary) Elgot algebras and
    ffg-Elgot algebras are the same concept. However, in the present
    setting this may not be the case.
    
  \item\label{R:app:2} Since fp-equations have variable objects $X$ such that
    $\C(X,-)$ preserves filtered colimits, one could expect that
    ffg-equations will have $X$ as those objects for which $\C(X,-)$
    preserves sifted colimits. Indeed, that would yield the same
    colimit $\phi F$, as we prove in the Appendix. 
  \item We do not know whether, for ffg-Elgot algebras, Weak
    Functoriality implies Functoriality.\smnote{Do we have a
      counterexample for this?} The proofs of our main results (in
    particular \autoref{prop:phi_is_elgot} and
    \autoref{thm:initialfy}) do not work when Weak Functoriality is
    replaced by Functoriality.
  \end{enumerate}
\end{rem}
\begin{rem}\label{rem:eff}
  In the case where $F\colon \Set^T \to \Set^T$ is a lifting of a
  functor $F_0\colon \Set \to \Set$ (via a distributive law
  $\lambda$), an $F$-algebra is given by a set $A$ equipped with both
  a $T$-algebra structure $\alpha\colon TA \to A$ and an $F_0$-algebra
  structure $a\colon F_0 A \to A$ such that $a$ is a $T$-algebra
  homomorphism, i.e.~one has
  $\alpha \o Ta = a \o F\alpha \o \lambda_A$. Morphisms of
  $F$-algebras are those maps that are both $T$-algebra and
  $F_0$-algebra homomorphisms. Now one may think of ffg-equations and
  their solutions as modelling \emph{effectful iteration}. Indeed, let
  $X_0$ be a finite set of variables and consider any map
  \[
    e_0\colon X_0 \to T(F_0 X_0 + A).
  \]
  This may be regarded as a system of recursive equations with
  variables from $X_0$ and parameters in $A$, where for every recursive
  call a side effect in $T$ might happen. If $(A, \alpha, a)$ is an
  $F$-algebra, a solution of such a recursive
  system should assign to each variable in $X_0$ an element of $A$,
  i.e.~we have a map $\sol e_0\colon X_0 \to A$, such that the square below
  commutes (here we write $+$ for disjoint union):
  \[
    \xymatrix@C+2pc@R-1pc{
      X_0 \ar[r]^-{\sol e_0}
      \ar[dd]_{e_0}
      &
      A 
      \\ 
      & 
      TA 
      \ar[u]_{\alpha}
      \\
      T(F_0 X_0 + A)
      \ar[r]_-{T(F_0 \sol e_0 + A)}
      &
      T(F_0A+A) \ar[u]_{T[a,A]}
    }
  \]
  Indeed, from $e_0$ we may form the map
  \[
    \ol e = (X_0 
    \xrightarrow{e_0} 
    T(F_0 X_0 + A)
    \xrightarrow{\cong}
    TF_0 X_0 \oplus TA
    \xrightarrow{\lambda_X \oplus \alpha}
    F TX_0 \oplus A),
  \]
  where $\oplus$ denotes the coproduct in $\C$, which may be different
  from disjoint union. 
  Then its unique extension $TX_0 \to FTX_0 \oplus A$ to a $T$-algebra
  morphism is an ffg-equation, and a solution $TX_0 \to A$ of this in
  the sense of \autoref{def:eqsol} is precisely the same as an
  extension of a solution for $e_0$ in the above sense. 
\end{rem}
\begin{construction}\label{constr:defsol}
  We aim at proving that $\phi F$ is the initial ffg-Elgot algebra. For
  that we first construct a solution $\sol e\colon X \to \phi F$ for every 
  ffg-equation $e\colon X\to FX+\phi F$.
  Recall that $\phi F = \colim D$ for the inclusion
  $D\colon \coafr F\monoto \coa F$ and denote the colimit injections
  by $\inj c: C \to \phi F$ for every ffg-coalgebra $(C,c)$.
  Thus $FX+\phi F = \colim (FX+D)$ with colimit injections
  $FX+c^\sharp$. Since $X$ is an ffg-object, this sifted colimit is
  preserved by $\C(X,-)$. Thus, the diagram
  \[
    \hat D\colon \coafr F\to \Set,\quad (C\xra{c}FC)\mapsto \C(X,FX+C)
  \]
  has
  \[ \colim \hat D = \C(X,FX+\phi F) \]
  with colimit injections given by postcomposition with $FX+c^\sharp$.

  By \autoref{R:new}\ref{R:new:1}, every ffg-equation
  $e\colon X \to FX + \phi F$ thus factorizes through one of the
  colimit injections $FX + \inj c$, i.e.~for some ffg-coalgebra
  $c\colon C \to FC$ and $w\colon X \to FX + C$ we have the
  commutative triangle below:
  \begin{equation}\label{eq:fac}
    \vcenter{
    \xymatrix{
      X 
      \ar[r]^-e 
      \ar[rd]_w
      & 
      FX + \phi F
      \\
      & FX + C\ar[u]_{FX + \inj c}
      }}
  \end{equation}
  We see that $w$ is an ffg-equation. We combine it with the
  ffg-equation $c$ (having the initial object $0$ as parameter, see
  \autoref{def:eqsol}) to $w \sq c\colon X + C \to F(X+C)$, which is
  an object of $\coafr F$. Finally, we put
  \begin{equation}\label{eq:defsol}
    \sol e = 
    (\,X 
    \xrightarrow{\inl} 
    X + C 
    \xrightarrow{\inj{(w \sq c)}}
    \phi F\,).
  \end{equation}
\end{construction}
We prove below that $\sol e$ is indeed a solution of $e$ in the
algebra $\phi F$ (cf.~\eqref{eq:sol}) and verify some properties used later. 
\begin{lemma}\label{lem:solution}
  The definition of $\sol e$ in~\eqref{eq:defsol} is
  independent of the choice of the factorization~\eqref{eq:fac}, and  $\sol e$ is a
  solution of $e$ in $\phi F$.
\end{lemma}
\proof
  \begin{enumerate}
  \item We first show the independence: given another ffg-coalgebra
    $\ol c\colon \ol C\to F\ol C$ and a factorization
    $e = (FX+\ol c^\sharp)\o \ol w$, we prove
    \begin{equation}\label{eq:fact}
      (w\sq c)^\sharp\o \inl = (\ol w \sq \ol c)^\sharp\o \inl.
    \end{equation}
    Recall the category $\el \hat D$ of elements of $\hat D$: its
    objects are triples $(C,c,w)$ where $(C,c)\in \coafr F$ and
    $w\in \hat D(C,c)$, i.e. $w\colon X\to FX+C$, and a morphism into
    $(\ol C, \ol c, \ol w)$ is a coalgebra homomorphism
    $h\colon (C,c)\to (\ol C,\ol c)$ with $(FX+h)\o w = \ol w$.

    Given two factorizations
    $(FX+c^\sharp)\o w = e = (FX+\ol c^\sharp)\o \ol w$, we thus see
    that the colimit injection $FX+c^\sharp$ takes the element $w$ to
    the same value to which the colimit injection $FX+\ol c^\sharp$
    takes $\ol w$. This implies that $w$ and $\ol w$ lie in the same
    connected component of $\el\hat D$. Therefore it suffices to prove
    \eqref{eq:fact} under the assumption that a morphism $h$ from $w$
    to $\ol w$ exists in $\el \hat D$: then that equation holds in the
    whole connected component. Thus, we have the following commutative
    diagram:
\[
\xymatrix{
  &
  X \ar[dl]_ w \ar[dr]^{\ol w}
  \\
  FX+C \ar[rr]^{FX+h} \ar[d]_{FX+c}
  &&
  FX+\ol C \ar[d]^{FX+\ol c}
  \\
  FX+FC \ar[rr]_{FX+Fh}
  &&
  FX+F\ol C
}
\] 
It follows that $X+h$ is a coalgebra homomorphism from $w\sq c$ to $\ol w\sq \ol c$. Indeed, in the following diagram
\[
\xymatrix@C+1em{
  X+C \ar[r]^-{[w,\inr]} \ar[d]_{X+h}
  &
  FX+C \ar[r]^-{FX+c} \ar[d]_{FX+h}
  &
  FX+FC \ar[r]^-\can \ar[d]_{FX+Fh}
  &
  F(X+C) \ar[d]^{F(X+h)}
  \ar@{<-} `u[l] `[lll]_{w\sq c} [lll]
  \\
  X+\ol C \ar[r]_-{[\ol w,\inr]}
  &
  FX+\ol C \ar[r]_-{FX+\ol c}
  &
  FX+F\ol C \ar[r]_-\can
  &
  F(X+\ol C)
  \ar@{<-} `d[l] `[lll]^-{\ol w \sq \ol c} [lll]
}
\]
the left-hand square and the middle one commute by the preceding
diagram, and the right-hand square commutes trivially. Since the
colimit injections $(\dash)^\sharp$ form a compatible family, we obtain
$\inj{(w \sq c)} = \inj{(\ol w \sq \ol c)} \cdot (X + h)$.
Precomposed with $\inl$ this yields the desired equation \eqref{eq:fact}.

\item We show that $e^\dag$ is a solution of $e$ in $\phi F$.

(2a) First note that the following triangle commutes:
\begin{equation}\label{eq:wctriangle}
  \vcenter{
  \xymatrix{
    C \ar[r]^{c^\sharp} \ar[d]_{\inr} & \phi F \\
    X+C \ar[ur]_{(w\sq c)^\sharp}
  }}
\end{equation}
To this end, we just need to verify that $\inr$ is a morphism in $\coafr F$ from $(C,c)$ to $(X+C,w\sq c)$, which is established by the commutative diagram below:
\[
\xymatrix{
C \ar[rrr]^c \ar[dd]_\inr \ar[dr]^{\inr} &&& FC \ar[dd]^{F\inr} \ar[dl]^{\inr} \\
& FX+C \ar[r]^{FX+c} & FX+FC \ar[dr]^{\can} &\\
X+C \ar[ur]^{[w,\inr]} \ar[rrr]_{w\sq c} & & & F(X+C)
}
\]
(2b) The commutative triangle \eqref{eq:wctriangle} together with
$(w\sq c)^\sharp\o \inl = e^\dag$ yield the following commutative
triangle:
\begin{equation}\label{diag:copair}
  \vcenter{
    \xymatrix@C+2em{
      FX+FC \ar[r]^{[Fe^\dag,Fc^\sharp]} \ar[d]_\can & F(\phi F)\\
      F(X+C) \ar[ur]_{F(w\sq c)^\sharp} & 
    }}
\end{equation}
We conclude that the following diagram
\begin{equation}\label{diag:upper-part}
  \vcenter{
    \xymatrix{
      X \ar[rr]^{e^\dag} \ar[d]_w \ar[dr]_\inl && \phi F \\
      FX+C \ar[dd]_{FX+c} & X+C \ar[d]^{w\sq c} \ar[ur]_{(w\sq c)^\sharp} & \\
      & F(X+C) \ar[dr]^{F(w\sq c)^\sharp} & \\
      FX+FC \ar[ur]_{\can} \ar[rr]_{[Fe^\dag, Fc^\sharp]} & & F(\phi F) \ar[uuu]_{t}
    }}
\end{equation}
commutes: the left-hand part follows from the definition of $w\sq c$,
the upper one is the definition of $e^\dag$, the right-hand one uses
that $(w\sq c)^\sharp$ is a coalgebra homomorphism, and the lower one is
the triangle~\eqref{diag:copair}.

We are ready to prove that $e^\dag$ is a solution of $e$, which means
that the outside of the following diagram commutes:
\[
  \xymatrix@C+1em{
    X \ar[rr]^-{e^\dag} \ar[d]_w
    & & \phi F
    \\
    FX+C \ar[r]^-{FX+c} \ar[d]_{FX+c^\sharp}
    &
    FX+FC \ar[d]_{FX+Fc^\sharp}
    \ar[r]^-{[Fe^\dag,Fc^\sharp]}
    &
    F(\phi F) \ar[u]_{t}
    \\
    FX+\phi F
    \ar@{<-} `l[u] `[uu]^e [uu]
    & 
    FX+F(\phi F)
    \ar[l]^-{FX+t}
    \ar[ur]_(.7)*+{\labelstyle [Fe^\dag,F(\phi F)]}
    \ar[r]_-{F\sol e + t}
    &
    F(\phi F)+\phi F
    \ar `r[u] `[uu]_{[t,\phi F]} [uu]
    \ar@{<-} `d[l] `[ll]^-{Fe^\dag+\phi F} [ll]
}
\]
The upper part has just been established in~\eqref{diag:upper-part}. The
left-hand part commutes by~\eqref{eq:fac}, the lower left-hand square
commutes because $c^\sharp$ is a coalgebra homomorphism, and the three
remaining parts commute trivially.\qed
\end{enumerate}
\doendproof
\begin{proposition}\label{prop:phi_is_elgot}
  The algebra $t\colon F(\phi F) \to \phi F$ together with the solution
  operator $\dagger$ from \autoref{constr:defsol} is an ffg-Elgot
  algebra.
\end{proposition}
\proof
  \emph{Weak Functoriality.} Suppose that the commutative square below
  and a morphism $h\colon Z\to \phi F$ are given, where $X$, $Y$, and
  $Z$ are ffg objects.
  \[
    \xymatrix{
      X \ar[r]^-e \ar[d]_m & FX + Z \ar[d]^{Fm + Z} \\
      Y \ar[r]_-f & FY + Z
    }
  \]
  Since $Z$ is ffg, the morphism $h$ factorizes through the colimit
  injection $c^\sharp$ of some coalgebra $c\colon C \to FC$ in
  $\coafr F$ as in the triangle below:
  \[
    \xymatrix{
      Z \ar[rd]_{v_0} \ar[r]^-h &  \phi F \\ 
      & C \ar[u]_{c^\sharp} 
    }
  \]
Form the two ffg-equations
\[
  v = v_0\bullet e\colon X\to FX+C \quad \text{and}\quad w =
  v_0\bullet f\colon Y \to FY+C,
\]
and observe that the following diagram commutes:
\[
\xymatrix@C+1em{
  X \ar[r]^-e \ar[d]_m
  &
  FX+Z \ar[r]^-{FX+v_0}
  \ar[d]^{Fm+Z}
  &
  FX+C \ar[d]^{Fm+C}
  \ar@{<-} `u[l] `[ll]_-v [ll]
  \\
  Y \ar[r]_-f
  &
  FY+Z \ar[r]_-{FY+v_0}
  &
  FY+C
  \ar@{<-} `d[l] `[ll]^-w [ll]
}
\]
Consequently, in the following diagram
\[
\xymatrix{
  X+C \ar[r]^-{[v,\inr]}
  \ar[d]_{m+C}
  &
  FX+C \ar[r]^-{FX+c}
  \ar[d]_{Fm+FC}
  & FX+FC \ar[r]^-\can
  \ar[d]^{Fm+FC}
  &
  F(X+C) \ar[d]^{F(m+C)}
  \ar@{<-} `u[l] `[lll]_-{v\sq c} [lll]
  \\
  Y+C \ar[r]_-{[w,\inr]}
  &
  FY+C \ar[r]_-{FY+c}
  &
  FY+FC \ar[r]_-\can & F(Y+C)
  \ar@{<-} `d[l] `[lll]^-{w\sq c} [lll]
}
\]
the left-hand square commutes. The other parts are clearly
commutative, and thus we see that $m+C$ is a coalgebra homomorphism
from $v\sq c$ to $w\sq c$. Therefore
\[
  (v\sq c)^\sharp = (w\sq c)^\sharp\o (m+C),
\]
which yields the desired equation $(h\bullet f)^\dag\o m = (h\bullet e)^\dag$, as shown by the commutative diagram below:
\[
\xymatrix@R-2em@C+1em{
  X \ar `u[r] `[rrd]^-{(h\bullet e)^\dag} [rrd]
  \ar[dd]_m
  \ar[r]^-{\inl}
  &
  X+C \ar[rd]^-{(v\sq c)^\sharp} \ar[dd]^{m+C}
  \\
  &&
  \phi F
  \\
  Y \ar[r]_-{\inl}
  \ar `d[r] `[rru]_-{(h\bullet f)^\dag} [rru]
  &
  Y+C \ar[ru]_-{(w\sq c)^\sharp}
}
\]
\emph{Compositionality.}
\begin{enumerate}
\item Suppose that two ffg-equations $e\colon X\to FX+Y$ and
  $f\colon Y\to FY+\phi F$ are given, and factorize $f$ through some
  colimit injection $FY+c^\sharp$ of $FY+C$:
  \[
    \xymatrix{
      Y \ar[rd]_v \ar[r]^-f & FY+\phi F \\
      & FY+C \ar[u]_{FY+c^\sharp} 
    }
  \]
  Then, by the definition of $\dag$, we have
  \[
    f^\dag = (v\sq c)^\sharp \o \inl.
  \]
  This implies that the ffg-equation
  $f^\dag\bullet e\colon X\to FX+\phi F$ factorizes as follows:
  \[
    \xymatrix{
      X \ar[r]^-e
      \ar[rrd]_-{\inl\bullet e}
      \ar `u[r] `[rr]^-{f^\dag\bullet e} [rr]
      &
      FX+Y \ar[r]^{FX+f^\dag} \ar[rd]^{FX+\inl}
      &
      FX+\phi F
      \\
      && FX+Y+C \ar[u]_{FX+(v\sq c)^\sharp} 
    }
  \]
  Thus, by the definition of $\dag$ again, the solution
  $(f^\dag\bullet e)^\dag\colon X\to \phi F$ of $f^\dag\bullet e$ is
  given by the coproduct injection $\inl\colon X\to X+Y+C$ followed by
  the colimit injection
  \[
    [(\inl\bullet e)\sq (v\sq c)]^\sharp\colon X+Y+C\to \phi F.
  \]
  By Remark \ref{rem:square_bullet_properties}(3) the last morphism is equal to
  $[e\sq (v\sq c)]^\sharp$, thus we obtain:
  \[
    (f^\dag\bullet e)^\dag = (X \xrightarrow{\inl} X+Y+C
    \xrightarrow{[e\sq (v\sq c)]^\sharp} \phi F).
  \]
  \item The equation $e\sq f\colon X+Y\to F(X+Y)+\phi F$ factorizes as follows:
    \[  
      \xymatrix@C+1em{
        X+Y \ar[r]^-{[e,\inl]} 
        \ar@(d,l) [ddrrr]_{e\sq v} 
        &
        FX+Y \ar[r]^-{FX+f} \ar[dr]_{FX+v}
        &
        FX+FY+\phi F \ar[r]^{\can+\phi F}
        &
        F(X+Y)+\phi F
        \ar@{<-} `u[l] `[lll]_-{e\sq f} [lll]
        \\
        & & FX+FY+C \ar[u]_{FX+FY+c^\sharp} \ar[dr]^{\can+C}
        \\
        &&&
        F(X+Y)+C \ar[uu]^{F(X+Y)+c^\sharp}  
      }
    \]
    Therefore, by the definition of $\dag$, we have 
    \[
      (e\sq f)^\dag = (
      X+Y \xrightarrow{\inl} X+Y+C \xrightarrow{[(e\sq v)\sq c]^\sharp}
      \phi F).
    \]
    Precomposing this with the coproduct injection $\inl\colon X\to X+Y$
    proves the desired equality
    \[
      (e\sq f)^\dag\o \inl = [(e\sq v)\sq c]^\sharp \o \inl = (f^\dag\bullet
      e)^\dag.\tag*{\qed}
    \]
  \end{enumerate}
\doendproof
\begin{defn}
  A \emph{morphism of ffg-Elgot algebras} from $(A,a,\dagger)$ to
  $(B,b,\ddagger)$ is a morphism $h\colon A \to B$ in $\C$
  \emph{preserving solutions}, i.e.~for every ffg-equation
  $e\colon X \to FX + A$ we have
  \[
    (h \aft e)^\ddagger = h\o e^\dag.
  \]
\end{defn}
Identity morphisms are clearly ffg-Elgot algebra morphisms, and
morphisms of ffg-Elgot algebra compose. Therefore ffg-Elgot algebras form
a category, which we denote by
\[
  \ffgElgot F.
\]
\begin{lemma}\label{lem:solpres}
  Morphisms of ffg-Elgot algebras are $F$-algebra homomorphisms.
\end{lemma}
\begin{proof}
  This is completely analogous to the proof
of~\cite[Lemma~4.2]{amv_elgot}. The only small modification is needed
at the beginning of the proof as follows:

\medskip
Let $\C_\ffg\mathord{/}A$ be the slice category of all arrows $q\colon X\to A$ with $X$ ffg. Since $\C$ is a variety, $A$ is the sifted colimit of the diagram $D_A\colon \C_\ffg\mathord{/}A\to \C$ given by $(q\colon X\to A)\mapsto X$.

\medskip
The remainder of the proof is identical. 
\end{proof}
Note that the converse of the above lemma fails in general. In fact, \cite[Example~4.4]{amv_elgot}
exhibits an \mbox{(ffg-)Elgot} algebra for the identity functor on $\Set$ 
and an algebra homomorphism on it which is not solution-preserving.  
\begin{theorem}\label{thm:ini}
  The triple $(\phi F, t, \dagger)$ is the initial
  ffg-Elgot algebra for $F$. 
\end{theorem}
\proof
  Let $(A,a,\ddag)$ be an ffg-Elgot algebra. For the initial object
  $0$ we denote by $i_A\colon 0\to A$ the unique morphism.
  \begin{enumerate}
  \item We obtain a cocone of the diagram
    \[ \coafr F \monoto \coa F\xra{U} \C, \] where $U$ is the
    forgetful functor, as follows: to every ffg-coalgebra
    $c\colon C\to FC$ assign the solution
    \[ (i_A\bullet c)^\ddag\colon C\to A \] of the ffg-equation
    $i_A\bullet c\colon C\to FC+A$. Indeed, given a coalgebra
    homomorphism $m\colon (C,c) \to (C',c')$ in $\coafr F$,
    Weak Functoriality applied to $h=i_A$ yields
    \[
      (i_A\bullet c)^\ddag = (C \xrightarrow{m} C'
      \xrightarrow{(i_A\bullet c')^\ddag} A).
    \]
    Since $\phi F$ is the colimit of the embedding
    $\coafr F\monoto \coa F$ and since $U$ preserves colimits, there
    exists a unique morphism $h\colon \phi F \to A$ in $\C$ such that
    the following triangles
    \[
      \xymatrix{
        C \ar[d]_{c^\sharp} \ar[dr]^{(i_A\bullet c)^\ddag}  \\
        \phi F \ar@{-->}[r]_h & A
      }
    \]
    commute for all ffg-coalgebras $c\colon C\to FC$.

  \item We prove that $h$ is solution-preserving. Given an
    ffg-equation $e\colon X\to FX+\phi F$, factorize $e$ through one
    of the colimit injections $FX+c^\sharp$ of $FX+\phi F$:
    \[
      \xymatrix{
        X \ar[rd]_v \ar[r]^-e & FX+\phi F \\
        & FX+C \ar[u]_{FX+c^\sharp} 
      }
    \]
    Since $e=c^\sharp\bullet v$, Remark
    \ref{rem:square_bullet_properties}(1) and the definition of $h$
    yield
    \[
      (h\bullet e)^\ddag
      =
      [ h\bullet (c^\sharp\bullet v) ]^\ddag
      =
      [ (h\o c^\sharp)\bullet v ]^\ddag
      =
      [ (i_A\bullet c)^\ddag \bullet v ]^\ddag.
    \]
    The last morphism is, due to Compositionality, equal to
    \[
      [ v\sq (i_A\bullet c) ]^\ddag \o \inl.
    \]
    Thus, it remains to verify that $h\o e^\dag$ is the same
    morphism. From $e=c^\sharp\bullet v$ the definition of $\dag$ yields
    $e^\dag = (v\sq c)^\sharp\o \inl$ and we get
    \[
      h\o e^\dag
      =
      h \o (v\sq c)^\sharp\o \inl
      =
      [ i_A\bullet (v\sq c)]^\ddag \o \inl
      =
      [ v\sq (i_A\bullet c) ]^\ddag\o \inl,
    \]
    where the last step uses
    Remark~\ref{rem:square_bullet_properties}(2).

  \item It remains to prove the uniqueness of $h$. Thus suppose that
    another solution-preserving morphism $g\colon \phi F\to A$ is
    given. It is sufficient to prove
    \[
      g\o c^\sharp = h\o c^\sharp \quad\text{for all ffg-coalgebras $c\colon
        C\to FC$.}
    \]
    Form the ffg-equation
    $i_{\phi F}\bullet c = \inl \o c\colon C\to FC+\phi F$. Then it is
    easy to verify that the left coproduct injection
    $\inl\colon C\to C+C$ is a coalgebra homomorphism from $(C,c)$ to
    $(C+C,\ol c)$ where $\ol c = (\inl\o c)\sq c$.
    Therefore, the compatibility of the colimit injections
    $(\dash)^\sharp$ yields $c^\sharp = \ol c^\sharp\o \inl$.
    Now $i_{\phi F}\bullet c$ factorizes as follows:
    \[ 
      \xymatrix@C+2em{
        C \ar[dr]_{\inl\o c} \ar[r]^-{i_{\phi F}\bullet c} & FC+\phi F \\
        & FC+C \ar[u]_{FC+c^\sharp}
      }
    \]
    Therefore the definition of $\dag$ yields
    \[
      (i_{\phi F}\bullet c)^\dag
      =
      ((\inl\o c)\sq c)^\sharp \o \inl
      =
      \ol c^\sharp \o \inl
      =
      c^\sharp. 
    \]
    Since $g$ preserves solutions, using Remark
    \ref{rem:square_bullet_properties}(1), and that 
    $g\o i_{\phi F}= i_A\colon 0\to A$, we thus get
    \[
      g\o c^\sharp
      =
      g\o (i_{\phi F}\bullet c)^\dag
      =
      (g\bullet (i_{\phi F}\bullet c))^\ddag
      =
      ((g\cdot i_{\phi F})\bullet c))^\ddag
      =
      (i_A\bullet c)^\ddag = h\o c^\sharp
    \]
    as required. This concludes the proof.\qed
  \end{enumerate}
\doendproof

The following result is the key to constructing free ffg-Elgot
algebras. In the case where $\C_\ffg = \C_\fp$, hence where ffg-Elgot
algebras agree with ordinary ones, we thus obtain a new result about
ordinary Elgot algebras.

\begin{theorem}\label{thm:initialfy}
Let $a\colon FA\to A$ be an $F$-algebra, $Y$ a free object of $\C$,
and $h\colon Y\to A$ a morphism. Then there is a bijective correspondence between
\begin{enumerate}
\item solution operators $\dag$ such that $(A,a,\dag)$ is an ffg-Elgot algebra for $F$, and
\item solution operators $\ddag$ such that $(A,[a,h],\ddag)$ is an ffg-Elgot algebra for $F(-)+Y$.
\end{enumerate}
\end{theorem}
\begin{rem}\label{rem:passages} The correspondence is given as follows:
\begin{enumerate}
\item For every ffg-Elgot algebra $(A,a,\dag)$ for $F$, we define a
  solution operator $\ddag$ w.r.t. $F(-)+Y$ as follows. Given $e\colon
  X\to FX+Y+A$, put
  \begin{equation}\label{eq:eh}
    e_h = (X\xra{e} FX+Y+A \xra{FX+[h,A]} FX+A)
  \end{equation}
  and 
  \[ e^\ddag \;:=\; e_h^\dag. \]  
\item Conversely, for every ffg-Elgot algebra $(A,[a,h],\ddag)$ for
  $F(-)+Y$, we define a solution operator $\dag$ w.r.t. $F$ as
  follows. Given an ffg-equation $e\colon X\to FX+A$, put
  \begin{equation}\label{eq:barr}
    \ol e = (X\xra{e} FX+A \xra{[\inl, \inr]} FX+Y+A)
  \end{equation}
  and 
  \[ e^\dag \;:=\; \ol e^\ddag. \]
\end{enumerate}
We will show that these two constructions are mutually inverse and yield the desired
bijective correspondence.
\end{rem}
In the next two subsections we will present the proof of
Theorem~\ref{thm:initialfy}. We will establish this result in two
steps: first we prove it for ffg objects $Y$ and then, using the first
step, for arbitrary free objects. Readers who would like to skip the
proof on first reading could jump straight to \autoref{S:free}.

\subsection{Proof of Theorem~\ref{thm:initialfy} for the case where $Y$ is an ffg object}\label{S:partA}
Suppose that $Y$ is an ffg object.
\begin{enumerate}
\item We prove that $(A,[a,h],\ddag)$ is an ffg-Elgot algebra whenever $(A,a,\dag)$ is.

\medskip (1a) Given an ffg-equation $e\colon X\to FX+Y+A$, then $e^\ddag$ is a solution, as shown by the diagram below:
\[
  \xymatrix{
    X \ar[rrr]^{e^\ddag=e_h^\dag} \ar[dd]_e \ar[dr]^{e_h}
    & & &
    A
    \\
    &
    FX+A \ar[r]^{Fe_h^\dag+A}
    &
    FA+A \ar[ur]^{[a,A]}
    \\
    FX+Y+A  \ar[ur]_(.6)*+{\labelstyle FX+[h,A]} \ar[rrr]_{Fe^\ddag + Y +A}
    & & &
    FA+Y+A \ar[ul]^(.6)*+{\labelstyle FA+[h,A]} \ar[uu]_{[[a,h],A]}
  }
\]
(1b) $\ddag$ is weakly functorial. Suppose that a commutative square
\[
\xymatrix{
X \ar[r]^-e \ar[d]_m & FX+Y+Z \ar[d]^{Fm+Y+Z} \\
X' \ar[r]_-{f} & FX'+Y+Z
} 
\]
and a morphism $g\colon Z\to A$ are given where $X$, $X'$ and $Z$ are
ffg objects. We need to prove
\[ (g\bullet e)^\ddag = (g\bullet f)^\ddag\o m. \]
From the following diagram:
\[
\xymatrix@1{
  X \ar[r]^-e
  &
  FX+Y+Z \ar[rr]^-{FX+Y+g} 
  &&
  FX+Y+A \ar[rr]^-{FX+[h,A]}
  &&
  FX+A
  \ar@{<-} `d[l] `[llll]^-{FX+[h,g]} [llll]
}
\]
we deduce
\[ (g\bullet e)_h = [h,g]\bullet e. \]
Here, by abuse of notation, $\bullet$ is used both for $F$ and $F(\dash)+Y$. Analogously, 
\[ (g\bullet f)_h = [h,g]\bullet f. \]
Since $\dag$ is weakly functorial, we get
\[ ([h,g]\bullet e)^\dag = ([h,g]\bullet f)^\dag\o m \]
and therefore
\[ (g\bullet  e)^\ddag = (g\bullet e)_h^\dag = ([h,g]\bullet e)^\dag = ([h,g]\bullet f)^\dag\o m = (g\bullet f)_h^\dag\o m = (g\bullet f)^\ddag\o m. \]
(1c) $\ddag$ is compositional. Given ffg-equations for $F(\dash)+Y$
\[ e\colon X\to FX+Y+Z\quad\text{and}\quad f\colon Z\to FZ+Y+A, \]
we are to prove
\[
(f^\ddag\bullet e)^\ddag = (e\sq f)^\ddag \o \inl.
\]
Express $A$ as a sifted colimit $a_i\colon A_i\to A$ ($i\in I$) of ffg
objects. Then also the morphisms $FZ+Y+a_i\colon FZ+Y+A_i\to FZ+Y+A$
form a sifted colimit cocone, and since $Z$ is an ffg object, $f$
factorizes through one of them:
\[
\xymatrix{
Z \ar[r]^-f \ar[dr]_{f_0} & FZ+Y+A\\
& FZ+Y+A_i \ar[u]_{FZ+Y+a_i}
}
\]
Define ffg-equations $\hat f$ and $\hat f_0$ by the commutative diagrams below (where $\inm$ denotes the middle coproduct injection):
\[
\xymatrix{
Y \ar[r]^h \ar[d]_\inl & A \ar[d]^\inr \\
Y+Z \ar[r]^-{\hat f}  & F(Y+Z)+A\\
Z \ar[u]^\inr \ar[r]_-f  & FZ+Y+A \ar[u]_{F\inr + [h,A]}
}
\qquad\qquad
\xymatrix{
Y \ar[d]_\inl \ar `r[dr]^-\inm [dr]& \\
Y+Z \ar[r]^-{\hat f_0} & F(Y+Z)+Y+A_i\\
Z \ar[r]_-{f_0}  \ar[u]^\inr & FZ+Y+A_i \ar[u]_{F\inr + Y+A_i}
}
\]
Since $\dag$ is compositional, we have
\[
(e\sq \hat f)^\dag\o \inl = (\hat f^\dag\bullet e)^\dag.
\]
We now verify that $[\inl,\inr]\colon X+Z\to X+Y+Z$ is a coalgebra
homomorphism from $e\sq f_0$ to $e\sq \hat f_0$. (Here we again use
$\sq$ for both $F$ and $F(\dash)+Y$.) This is shown by the commutative
diagram below, where $\can$ in the upper row is w.r.t.~$F(-) + Y$, and
in the lower row it is w.r.t.~$F$:
\[
  \let\objectstyle=\labelstyle
  \xymatrix@C-.5pc{
    X+Z \ar[r]^-{[e,\inr]} \ar[d]^{[\inl,\inr]}
    & FX+Y+Z \ar[rr]^-{FX+Y+f_0} \ar@{=}[d]
    &&
    FX+Y+FZ+Y+A_i \ar[rr]^{\can + A_i}
    &&
    F(X+Z) + Y + A_i \ar[d]_{F[\inl,\inr]+Y+A_i}
    \ar@{<-} `u[l] `[lllll]_-{e\sq f_0} [lllll]
    \\
    X+Y+Z
    \ar[r]_-{[e,\inr]}
    &
    FX+Y+Z \ar[rr]_-{FX+\hat f_0}
    &&
    FX+F(Y+Z)+Y+A_i \ar[rr]_{\can+Y+A_i}
    &&
    F(X+Y+Z)+Y+A_i
    \ar@{<-} `d[l] `[lllll]^-{e\sq \hat  f_0} [lllll]
    \\
}
\]
Moreover, we have
\[
[h,a_i]\bullet (e\sq f_0) = (e\sq f)_h
\]
as shown by the following computation:
\begin{align*}\
[h,a_i]\bullet (e\sq f_0) &= ([h,A]\o (Y+a_i))\bullet (e\sq f_0) &\\
&= [h,A]\bullet ((Y+a_i)\bullet (e\sq f_0)) & \text{Remark \ref{rem:square_bullet_properties}(1)}\\
&= [h,A]\bullet (e\sq\, ((Y+a_i)\bullet f_0)) & \text{Remark \ref{rem:square_bullet_properties}(2)}\\
&= [h,A]\bullet (e\sq f) & \text{def. $f_0$}\\
&= (e\sq f)_h & \text{def. $(\dash)_h$.}
\end{align*}
Analogously, 
\[ 
[h,a_i]\bullet (e\sq \hat f_0) = e\sq \hat f.
\]
Since $\dag$ is weakly functorial, we get
\begin{equation}\label{eq:hatprop2}
(e\sq f)_h^\dag = (e\sq \hat f)^\dag\o [\inl,\inr].
\end{equation}
We apply the Weak Functoriality of $\dag$ also the to lower square of
the diagram defining $\hat f_0$ and to $[h,a_i]$ in lieu of $h$ and
use that $[h,a_i]\cdot \hat f_0 = \hat f$ to obtain
\[
([h,a_i]\bullet f_0)^\dag = ([h,a_i]\bullet \hat f_0)^\dag\o \inr = \hat f^\dag \o \inr.
\]
This implies that 
\[
\hat f^\dag \o \inr = f_h^\dag
\]
since, using Remark \ref{rem:square_bullet_properties}(2),
\[
\hat f^\dag \o \inr = ([h,a_i]\bullet f_0)^\dag = ( [h,A]\bullet ((Y+a_i)\bullet f_0 ) )^\dag = ( [h,A]\bullet f )^\dag = f_h^\dag.
\]
We conclude
\begin{equation}\label{eq:hatprop}
  \hat f^\dag = [h,f_h^\dag]\colon Y+Z\to A
\end{equation}
since the left-hand component $\hat f^\dag \o \inl = h$ follows from the fact that $\hat f^\dag$ is a solution of $\hat f$:
\[
  \xymatrix@-.7pc{
    Y+Z \ar[rr]^{\hat f^\dag} \ar[ddd]_{\hat f} & &  A\\
    & Y \ar[ul]_-\inl \ar[d]_-h \ar[ur]^-h \\
    & A \ar@{=}[uur] \ar[dl]_-\inr \ar[dr]^-\inr \\
    F(Y+Z)+A \ar[rr]_{F\hat f^\dag + A}  & &  FA+A \ar[uuu]_{[a,A]}
}
\]
Thus, we conclude the proof with the following computation:
\begin{align*}
(f^\ddag\bullet e)^\ddag &= (f_h^\dag\bullet e)_h^\dag & \text{def. $\ddag$}\\
&= ( [h,f_h^\dag]\bullet e )^\dag & \text{def. $(\dash)_h$}\\
&= (\hat f^\dag\bullet e)^\dag & \text{\eqref{eq:hatprop}}\\
&= (e\sq \hat f)^\dag \o \inl & \text{compositionality of $\dag$}\\
&= (e \sq f)_h^\dag \o \inl  &  \text{\eqref{eq:hatprop2}}\\
&= (e \sq f)^\ddag \o \inl & \text{def. $\ddag$}
\end{align*}

\item For every ffg-Elgot algebra $(A,[a,h]^\ddag)$ for $F(\dash)+Y$, we prove that $(A,a,\dag)$ with $e^\dag := \ol e^\ddag$ is an ffg-Elgot algebra for $F$.

\medskip (2a) $e^\dag$ is a solution of $e\colon X\to FX+A$:
\[
\xymatrix@C+2em{
  X \ar[r]^-{e^\dag = \ol e^\ddag}
  \ar[d]^{\ol  e}
  &
  A
  \\
  FX+Y+A \ar[r]_-{F\ol e^\ddag + Y +A}
  &
  FA+Y+A \ar[u]^{[[a,h],A]}
  \\
  FX+A
  \ar[u]_{FX+\inr}
  \ar[r]_-{Fe^\dag+A}
  \ar@{<-} `l[u] `[uu]^e [uu] 
  &
  FA+A
  \ar `r[u] `[uu]_{[a,A]} [uu]
  \ar[u]^{[\inl,\inr]}
}
\]
Indeed, the upper square commutes since $\ol e^\ddag$ is a solution of
$\ol e$, and for the lower one recall that $\sol e = {\ol e}^\ddag$.

\medskip (2b) $\dag$ is weakly functorial. Given a coalgebra
homomorphism $m$ from $e\colon X \to FX + Z$ to $f\colon X' \to FX'+Z$
and a morphism $h\colon Z\to A$ where $X$, $X'$, and $Z$ are ffg objects, we need to prove
$(h\bullet e)^\dag = (h\bullet f)^\dag \o m$.
From the following diagram we see that $m$ is also a coalgebra
homomorphism for $F(\dash)+Y+Z$ from $\ol e$ to $\ol f$:
\[
  \xymatrix@C+1em{
    X \ar[r]^-e \ar[d]_m
    &
    FX+Z \ar[d]^{Fm+Z} \ar[r]^-{[\inl,\inr]}
    &
    FX+Y+Z \ar[d]^{Fm+Y+Z}
    \ar@{<-} `u[l] `[ll]_-{\ol e} [ll]
    \\
    X' \ar[r]_-f
    &
    FX'+Z \ar[r]_-{[\inl,\inl]}
    &
    FX'+Y+Z
    \ar@{<-} `d[l] `[ll]^-{\ol f} [ll]
  }
\]
Hence, Weak Functoriality of $\ddag$ yields
\[ (h\bullet \ol e)^\ddag = (h\bullet \ol f)^\ddag \o m. \]
This implies the desired equality since
\begin{equation}\label{eq:barprop} \ol{h\bullet e} = h\bullet \ol e 
\end{equation}
(and analogously for $f$) due to the following diagram:
\[
\xymatrix{
  X \ar[r]^-{h\bullet e} \ar[dr]_e
  &
  FX+A \ar[r]^-{[\inl,\inr]}
  &
  FX+Y+A
  \ar@{<-} `u[l] `[ll]_-{\ol{h\bullet e}} [ll]
  \\
  &
  FX+Z \ar[r]_-{[\inl,\inr]} \ar[u]_{FX+h}
  &
  FX+Y+Z
  \ar[u]_{FX+Y+h}
  \ar@{<-} `d[l] `[ull] [ull]^(.4){\ol e}\\
}
\]
\medskip (2c) $\ddag$ is compositional. Given ffg-equations
$e\colon X\to FX+Z$ and $f\colon Z\to FZ+A$, we need to prove
$(f^\dag\bullet e)^\dag = (e\sq f)^\dag \o \inl$. We first observe that 
\begin{equation}\label{eq:sqprop}
\ol e \sq \ol f = \ol{e\sq f}.
\end{equation}
This follows from the diagram below (where $\can$ on the right-hand
arrow is w.r.t.~$F(-) + Y$ and $\can$ in the middle of the diagram w.r.t.~$F$):
\[
  \let\objectstyle=\labelstyle
  \xymatrix@C-1em{
    X+Z
    \ar[r]_(.65){\begin{turn}{-40}$\labelstyle\!\![e,\inr]$\end{turn}}
    \ar[ddrr]_-{e\sq f}
    &
    FX+Z
    \ar[r]_(.6){\begin{turn}{-35}$\labelstyle\!\![\inl,\inr]$\end{turn}}
    \ar[dr]_-{FX+f}
    &
    FX+Y+Z \ar[r]_(.25){\begin{turn}{30}$\labelstyle\/FX+Y+f\!\!\!$\end{turn}}
    \save+<-3pt,5pt>*{}\ar@{<-} `u[l] `[ll]_-{[\ol e, \inr]} [ll]\restore
    \save+<3pt,5pt>*{}\ar@{-<} `u[r] `[rr]^-{FX + Y + \ol f} [rr]\restore
    &
    FX+Y+FZ+A \ar[r]_(.25){\begin{turn}{30}$\labelstyle FX+Y+[\inl,\inr]$\end{turn}}
    &
    FX+Y+FZ+Y+A \ar[dd]^{\can+A}
    \\
    & &
    FX+FZ+A \ar[ur]_(.6)*+{\labelstyle [\inl,\inr]+A} \ar[d]^{\can+A}
    \\
    & &
    F(X+Z) +A \ar[rr]_-{[\inl,\inr]}  
    & &
    F(X+Z)+Y+A
    \ar@{<-} `d[l] `[llll] [lllluu]^(.4){\ol{e\sq f}} 
}
\]
Note that the upper path composed with $\can +A$ yields
$\ol {e\sq f}$. The proof of compositionality now easily follows:
\begin{align*}
(f^\dag\bullet e)^\dag &= \left(\ol {\ol f^\ddag \bullet e}\right)^\ddag & \text{def. $\dag$}\\
&= (\ol f^\ddag \bullet \ol e)^\ddag & \text{by \eqref{eq:barprop}}\\
&= (\ol e \sq \ol f)^\dag\o \inl & \text{$\ddag$ compositional}\\
&= (\ol {e\sq f})^\ddag \o \inl & \text{by \eqref{eq:sqprop}}\\
&= (e\sq f)^\dag \o \inl & \text{def. $\dag$}
\end{align*}

\item We prove that the two passages (1) and (2) in Remark~\ref{rem:passages} are mutually inverse.

  \medskip \noindent (3a) The fact that (2) followed by (1) yields the identity
  is easy to see since for every ffg-equation $e\colon X\to FX+A$ for $F$, we have
  \[ \ol e_h = e, \]
  as shown by the commutative diagram below:
  \[  
    \xymatrix{
      X
      \ar[r]^-e
      \ar `d[drr] [drr]_-e
      &
      FX+A \ar@{=}[dr] \ar[r]^-{\inl+A}
      &
      FX+Y+A
      \ar[d]^{FX+[h,A]}
      \ar@{<-} `u[l] `[ll]_-{\ol e} [ll]
      \\
      & &
      FX+A
    }
  \]
  \medskip (3b) In order to show that (1) followed by (2) is the
  identity, we prove for every ffg-equation $e\colon X\to FX+Y+A$
  that $\ol{(e_h)}^\ddag = e^\ddag$. (We do \emph{not} claim
  that $\ol{(e_h)}=e$.) Express $A$ as a sifted colimit
  $a_i\colon A_j\to A$ ($j\in J$) of ffg objects. Then also the
  morphisms $FX+Y+a_i\colon FX+Y+A_i\to FX+Y+A$ form a sifted colimit
  cocone, and since $X$ is an ffg object, there exists $j\in J$ and a morphism
  $e_0$ such that the following triangle commutes:
\[
\xymatrix{
X \ar[r]^-e \ar[dr]_{e_0} & FX+Y+A\\
& FX+Y+A_j \ar[u]_{FX+Y+a_j}
}
\]
Consider the ffg-equation
\[ f = (Y+A_j \xra{\inr} F(Y+A_j)+Y+A_j \xra{F(Y+A_j)+Y+a_j} F(Y+A_j)+Y+A). \]
(Note that $f=a_j\bullet f_0$ for $f_0 = \inr$.) We have that
\begin{equation}\label{eq:new}
  f^\ddag = (Y+A_j \xra{Y+a_j} Y+A \xra{[h,A]} A)
\end{equation}
as demonstrated by the diagram below:
\[
  \xymatrix@C+2em{
    Y+A_j \ar[r]^{f^\ddag} \ar[d]^{Y+a_j}
    &
    A
    \\
    Y+A \ar[d]^\inr \ar[ur]_{[h,A]} \ar[dr]_\inr
    \\
    F(Y+A_j)+Y+A
    \ar[r]_-{Ff^\ddag+Y+A}
    \ar@{<-} `l[u] `[uu]^f [uu]
    &
    FA+Y+A \ar[uu]_{[[a,h],A]}
}
\]
We also have that
\begin{equation}\label{eq:ehprop} 
\ol{(e_h)} = f^\ddag\bullet e_0',
\end{equation}
where
\[
  e_0' = (X \xra{e_0} FX+Y+A_j \xra{\inl+Y+A_j} FX+Y+Y+A).
\]
Indeed, the following diagram commutes using~\eqref{eq:new} for the
right-hand part and~\eqref{eq:eh},~\eqref{eq:barr} for the lower
left-hand one: 
\[
\xymatrix@C+3em{
  X \ar[r]^-{e_0} \ar[dr]^-e \ar[ddr]_-{e_h}
  &
  FX+Y+A_j  \ar[d]^{FX+Y+a_j} \ar[r]^{\inl+Y+A_j}
  &
  FX+Y+Y+A_j \ar[d]_{FX+Y+Y+a_j}
  \ar`r[d] `[dd]^{FX+Y+f^\ddag} [dd]
  \\
  &
  FX+Y+A \ar[r]^-{\inl+Y+A} \ar[d]^{FX+[h,A]}
  &
  FX+Y+Y+A \ar[d]_{FX+Y+[h,A]}
  \\
  &
  FX+A \ar[r]_-{[\inl, \inr]}
  &
  FX+Y+A
  \ar@{<-} `d[l] `[ll] [lluu]^{\ol{(e_h)}}  
}
\]
Finally, we have a coalgebra homomorphism $\inl$ from $e_0$ to $e_0'\sq f_0$:
\[
  \xymatrix@C+4em{
    X \ar[r]^-\inl \ar[dr]^-{e_0} \ar[dddd]_{e_0}
    &
    X+Y+A_j  \ar[d]^{[e_0,\inr]}
    \\
    & FX+Y+A_j \ar[d]^{\inl+Y+A_j}
    \\
    & FX+Y+Y+A_j \ar[d]^{FX+Y+f_0}
    \\
    &
    FX+Y+F(Y+A_j) + Y+A_j \ar[d]^{\can+A_j}
    \\
    FX+Y+A_j
    \ar@(ru,lu)[ruu]^-*+{\labelstyle\inl+Y+A_j}
    \ar[ru]^-*+{\labelstyle \inl +Y+A_j}
    \ar[r]_-{F\inl + Y + A_j}
    &
    F(X+Y+A_j)+Y+A_j
  }
\]
Thus, we obtain
\begin{align*}
e^\ddag &= (a_j\bullet e_0)^\ddag & \\
&= (a_j\bullet (e_0'\sq f_0))^\ddag \o \inl & \text{$\ddag$ weakly functorial} \\
&= (e_0' \sq\, (a_j\bullet f_0))^\ddag \o \inl & \text{Remark \ref{rem:square_bullet_properties}(2)} \\
&= (e_0' \sq f)^\ddag \o \inl & \text{def. $f_0$} \\
&= (f^\ddag\bullet e_0')^\ddag & \text{$\ddag$ compositional} \\
&= \ol{(e_h)}^\ddag & \text{by \eqref{eq:ehprop}}.
\end{align*}
This concludes the proof. 
\end{enumerate}
\qed

\subsection{Proof of Theorem~\ref{thm:initialfy} for an arbitrary
  free object}\label{S:partB}

Now assume that $Y$ is an arbitrary free object of $\C$. We shall
reduce this case to the previous situation using filtered colimits.

\begin{notation}\label{not:compfam} Fix an $F$-algebra
  $a\colon FA\to A$ and a morphism $h\colon Y\to A$. Since in every
  variety $\C$ the free functor (left adjoint to the forgetful functor
  from $\C$ to $\Set^S$) preserves colimits, we can express the free
  object $Y$ as a colimit of a filtered diagram $D_Y$ of ffg objects $Y_i$:
  \[
    Y = \colim Y_i\qquad
    \text{with injections $y_i\colon Y_i\to Y$} \quad (i\in I).
  \]
\end{notation}
\begin{defn}\label{def:compfamalg}
  By a \emph{compatible family of ffg-Elgot algebras} is meant a family
  \begin{equation}\label{eq:compfamalg}
    (A,[a, h_i], (\dash)^{\dag,i})\qquad\text{(for $i \in I$)}
  \end{equation}
  of ffg-Elgot algebras for the functors $F(\dash)+Y_i$ such that for
  every connecting morphism $y_{ij}\colon Y_i\to Y_j$ of the diagram
  $D_Y$ and every ffg-equation $e\colon X\to FX+Y_i+A$, one has
  \[
    ((FX+y_{ij}+A)\o e)^{\dag,j} = e^{\dag,i}.
  \]
\end{defn}
To establish Theorem \ref{thm:initialfy}, we prove the following more refined result:
\begin{theorem}\label{thm:corr}
  For every $F$-algebra $(A,a)$ there is a bijective correspondence between 
  \begin{enumerate}
  \item solution operations $\dag$ such that $(A,a,\dagger)$ is an
    ffg-Elgot algebra for $F$,
  \item families of solution operations $(\dash)^{\dag,i}$ such that
    $(A,[a,h_i],(\dash)^{\dag,i})$ ($i\in I$) is a compatible family
    of ffg-Elgot algebras, and
  \item solution operations $\ddagger$ such that $(A,[a,h],\ddagger)$
    is an ffg-Elgot algebra for $F(\dash)+Y$.
\end{enumerate}
\end{theorem}
The proof is split into four lemmas.
\begin{lemma}\label{lem:elgot_to_compfamalg}
  Let $(A,a,\dag)$ be an ffg-Elgot algebra.  Every cocone
  $h_i\colon Y_i\to A$ ($i\in I$) induces a compatible family of
  ffg-Elgot algebras $(A,[a,h_i], (\dash)^{\dag,i})$ with solution
  operations given by
  \[
    e^{\dagger,i}
    =
    (X\xra{e} FX+Y_i+ A \xra{FX+[h_i,A]} FX+A)^\dag.
  \]
\end{lemma}
\begin{proof}
  By part~(1) in Subsection~\ref{S:partA},
  $(A,[a,h_i], (\dash)^{\dag,i})$ is an ffg-Elgot algebra for every
  $i\in I$. For compatibility, let $e\colon X\to FX+Y_i+A$ be an
  ffg-equation and let $y_{ij}\colon Y_i\to Y_j$ be a connecting
  morphism of $D_Y$. Then the triangle below commutes:
  \[
    \xymatrix@C+2pc{
      FX+Y_i+A \ar[r]^{FX+[h_i,A]} \ar[d]_{FX+y_{ij}+A}
      &
      FX+A
      \\
      FX+Y_j + A \ar[ur]_(.6)*+{\labelstyle FX+[h_j,A]} 
    }
  \]
  Therefore
  \begin{align*}
    ((FX+y_{ij}+A)\o e)^{\dag,j}
    &= ((FX+[h_j,A])\o (FX+y_{ij}+A)\o e)^{\dag}\\
    &= ((FX+[h_i,A])\o e)^{\dag}\\
    &= e^{\dag,i}
  \end{align*}
  Here the first equation is the definition of $(\dash)^{\dag,j}$, the
  second one follows from the above commutative triangle, and the last
  one is the definition of $(\dash)^{\dag,i}$.
\end{proof}
\begin{lemma}\label{lem:compfamalg_to_elgot}
  Suppose that a compatible family \eqref{eq:compfamalg} of ffg-Elgot
  algebras is given. Then for every ffg equation $e\colon X \to FX +
  A$ the morphism
  \[
    e^\dag = (X\xra{e} FX+A \xra{[\inl,\inr]} FX+Y_i+A)^{\dag,i}
  \]
  is independent of the choice of $i$. Moreover, $(A,a,\dag)$ is
  an ffg-Elgot algebra for $F$, and the morphisms $h_i$ ($i\in I$) form a
  cocone of the diagram $D_Y$.
\end{lemma}

\begin{proof}
  (1) By part~(2) in Subsection~\ref{S:partA}, we know that
  $(A,a,\dag)$ is an ffg-Elgot algebra. Let us verify that $\dag$ is
  independent of the choice of $i$. Given $i,j\in I$, choose $k\in I$
  and connecting morphisms $y_{ik}\colon Y_i\to Y_k$ and
  $y_{jk}\colon Y_j\to Y_k$, using that $D_Y$ is filtered. Then the following diagram
  commutes:
  \[
    \xymatrix@C+2em{
      &
      FX+Y_i+A \ar[d]^{FX+y_{ik}+A}
      \\
      FX+A
      \ar[r]^-{[\inl,\inr]}
      \ar[ur]^-{[\inl,\inr]}
      \ar[dr]_-{FX+\inr}
      &
      FX+Y_k+A
      \\
      &
      FX +Y_j + A \ar[u]_{FX+y_{jk}+A}
}
\]
Therefore, by compatibility of the family \eqref{eq:compfamalg}, one has 
\[
  (X\xra{e} FX+A \xra{[\inl,\inr]} FX+Y_i+A)^{\dag,i}
  =
  (X\xra{e} FX+A \xra{[\inl,\inr]} FX+Y_j+A)^{\dag,j},
\]
as required.

\medskip\noindent
(2) Next, we show that for every $i\in I$ the ffg-equation $Y_i \xra{\inm} FY_i+Y_i+A$ has the solution $\inm^{\dag,i}=h_i$:
\[
\xymatrix@C+3em{
Y_i \ar[r]^{\inm^{\dag,i}} \ar[d]_{\inm} \ar[dr]^{\inm} & A\\
FY_i+Y_i+A \ar[r]_{F\inm^{\dag,i}+Y_i+A} & FA+Y_i+A \ar[u]_{[a,h_i,A]}\\
}
\]
Here the outside commutes by the definition of a solution, and the
lower triangle commutes trivially. Therefore the upper triangle
commutes, showing that $h_i = \inm^{\dag,i}$.

\medskip\noindent
(3) 
Finally, we prove that the $h_i$'s form a cocone. Suppose that a
connecting morphism $y_{ij}\colon Y_i\to Y_j$ is given, and consider
the following commutative diagram (here $i_A\colon 0 \to A$ denotes the
unique morphism from the initial object to $A$):
\[
  \xymatrix@C+2em{
    &&&
    FY_i+Y_i+A \ar[d]_{FY_i+y_{ij}+A}
    \\
    Y_i \ar[r]^-{\inm} \ar[d]^{y_{ij}}
    \ar `u[urrr] [urrr]^{\inm}
    &
    FY_i + Y_i + 0 
    \ar[r]^{FY_i+y_{ij}+0}
    &
    FY_i+Y_j+0  \ar[r]^{FY_i+Y_j+i_A} \ar[d]_{Fy_{ij}+Y_j+0}
    & FY_i+Y_j+A \ar[d]_{Fy_{ij}+Y_j+A}
    \\
    Y_j  \ar[rr]_-{\inm} \ar[rr]
    &&
    FY_j+Y_j+0 \ar[r]_{FY_j+Y_j+i_A}
    &
    FY_j+Y_j+A
    \ar@{<-} `d[l] `[lll]^-{\inm} [lll]
  }
\]
Then we get
\begin{align*}
h_i &= \inm^{\dag,i}\\
&= (i_A \aft ((FY_i + y_{ij}+0)\o \inm))^{\dag,j}\\
&=(i_A \aft \inm)^{\dag,j} \o y_{ij}\\
&= \inm^{\dag,j} \o y_{ij}\\
&= h_j\o y_{ij}
\end{align*}
Here the first equation follows from part (2) above, the second one
follows from the upper part of the above diagram and compatibility,
the third one follows from the central part of the diagram via Weak
Functoriality of $(\dash)^{\dag,j}$, the fourth one is the lower part
of the diagram, and the last equation is again part (2).
\end{proof}

\begin{lemma}\label{lem:elgotfy_to_compfamalg}
  Every ffg-Elgot algebra $(A,[a,h],\ddagger)$ for $F(\dash)+Y$
  induces the following compatible family of ffg-Elgot algebras:
  $(A,[a,h_i],(\dash)^{\dag,i})$ ($i\in I$), where $h_i = h\o y_i$ and
  the solution operations are given by
  \[
    e^{\dag,i} = (X\xra{e} FX + Y_i+A \xra{FX+y_i+A}
    FX+Y+A)^\ddagger.
  \]
\end{lemma}

\proof
  \begin{enumerate}
  \item We first show that $(A,[a,h_i],(\dash)^{\dag,i})$ is an
    ffg-Elgot algebra for every $i\in I$.  In the following, for every
    ffg-equation $e\colon X\to FX+Y_i+A$, we put
\[
  \ol e = (X\xra{e} FX + Y_i+A \xra{FX+y_i+A} FX+Y+A).
\]
\emph{Solution.} Consider the diagram below:
\[
  \xymatrix@C+2em{
    X \ar[r]^{e^{\dag,i}} \ar[d]^{\ol e}
    &
    A
    \\
    FX+Y+A \ar[r]^{Fe^{\ddag,i}+Y+A}
    &
    FA+Y+A \ar[u]^{[a,h,A]}
    \\
    FX+Y_i+A
    \ar[u]_{FX+y_i+A}
    \ar[r]_{Fe^{\dag,i}+Y_i+A}
    \ar@{<-} `l[u] `[uu]^e [uu]
    &
    FA+Y_i+A \ar[u]^{FA+y_i+A}
    \ar `r[u] `[uu]_{[a,h_i,A]} [uu]
  }
\]
The upper part commutes because $e^{\dag,i}=\ol e^\ddag$ is the
solution of $\ol e$, and the other three parts commute
trivially. Therefore the outside of the diagram commutes, showing that
$e^{\dag,i}$ is a solution of $e$.

\medskip
\emph{Weak functoriality.} Suppose that two ffg-equations $e\colon X
\to FX + Y_i + Z$ and $f\colon X' \to FX' + Y_i + Z$ are given
together with a coalgebra homomorphism $m$ from $e$ to $f$ and a
morphism $g\colon Z \to A$. Then $m$ is also a coalgebra homomorphism
w.r.t.~$F(-) + Y$:
\[
  \xymatrix{
    X \ar[r]^-e \ar[d]_m
    &
    FX + Y_i + Z \ar[rr]^-{FX + y_i + A}
    \ar[d]^{Fm + Y_i + A}
    &&
    FX + Y + A
    \ar[d]^{Fm + Y + A}
    \ar@{<-} `u[l] `[lll]_-{\ol e} [lll]
    \\
    X'
    \ar[r]_-f
    &
    FX' + Y_i + A \ar[rr]_-{FX' + y_i + A}
    &&
    FX' + Y + A
    \ar@{<-} `d[l] `[lll]^-{\ol f} [lll]
    }
\]
Moreover, we have
\begin{equation}\label{eq:gaftol}
  g \aft \ol e = \ol{g \aft e}
\end{equation}
and similarly for $f$, due to the following diagram:
\[
  \xymatrix{
    X
    \ar[r]^-e
    \ar[rd]_-{\ol e}
    &
    FX + Y_i + Z
    \ar[rr]^-{FX + Y_i + g}
    \ar[d]^{FX+y_i + Z}
    &&
    FX + Y_i + A
    \ar[d]^-{FX + y_i + A}
    \ar@{<-} `u[l] `[lll]_-{g \aft e} [lll]
    \\
    &
    FX + Y + Z
    \ar[rr]_-{FX + Y + g}
    &&
    FX + Y + A
    \ar@{<-} `d[l] `[lll]^-{g \aft \ol e} [lllu]
    }
\]
Thus, Weak Functoriality of $\dag,i$ follows from that of $\ddag$:
\[
  (g \aft f)^{\dag,i} \o m = (\ol{g \aft f})^\ddag \o m = (g \aft \ol
    f)^\ddag \o m = (g \aft \ol e)^\ddag  = (\ol{g \aft e})^\ddag = (g
    \aft e)^{\dag,i}.
\]

\medskip \emph{Compositionality.} Using the definition of
$(\dash)^{\dag,i}$, one easily verifies that for two ffg-equations
$e\colon X\to FX+Y_i+Z$ and $f\colon Z\to FZ+Y_i+A$ one has
$\ol{f \sq e} = \ol f \sq \ol e$ due to the following commutative
diagram:
\[
  \let\objectstyle=\labelstyle
  \xymatrix{
    X+Z
    \ar[r]^-{[e,\inr]}
    \ar[rd]_-{[\ol e, \inr]}
    &
    FX + Y_i + Z
    \ar[rr]^-{FX + Y_i +f}
    \ar[d]^{FX + y_i + Z}
    &&
    FX + Y_i + FZ + Y_i + A
    \ar[d]|*+{\labelstyle FX + y_i + FZ + y_i + A}
    \ar[r]^-{\can + A}
    &
    F(X+Z) + Y_i + Z
    \ar[d]^-{F(X+Z) + y_i + A}
    \ar@{<-} `u[l] `[llll]_-{f \sq e} [llll]
    \\
    &
    FX + Y + A
    \ar[rr]_-{FX + Y + \ol f}
    &&
    FX + Y + FZ + Y + A
    \ar[r]_-{\can + A}
    &
    F(X+Z) + Y + A
    \ar@{<-} `d[l] `[llll]^-{\ol f \aft \ol e} [llllu]
  }
\]
Thus we obtain $(f \sq e)^{\dag,i} = (\ol{f \sq e})^\ddag = (\ol f \sq
\ol e)^\ddag$, and we have
\[
  (f^{\dag,i} \aft e)^{\dag,i} = (\ol f^\ddag \aft e)^{\dag,i} =
  \left(\ol{\ol f^\ddag \aft e}\right)^\ddag = (\ol f^\ddag \aft e)^\ddag.
\]
Then compositionality of $\ddagger$ implies
\[ (f\sq e)^{\dag,i}\o \inl = (\ol f \sq \ol e)^\ddagger\o \inl = (\ol f^{\ddagger}\bullet \ol e )^\ddagger = (f^{\dag,i}\bullet e)^{\dag,i}.  \]

\item To prove that the given family of ffg-Elgot algebras is
  compatible,  let  $e\colon X\to FX+Y_i\to A$ be an ffg-equation and
  $y_{ij}\colon Y_i\to Y_j$ a connecting morphism of $D_Y$. Then 
  \begin{align*}
    ((FX+y_{ij}+A)\o e)^{\dag,j}
    &= ((FX+y_j+A)\o (FX+y_{ij}+A)\o e)^{\ddagger}
    \\
    &= ((FX+y_i+A)\o e)^{\ddagger}\\
    &= e^{\dag,i},
  \end{align*}
  where the first equation uses the definition of $(\dash)^{\dag,j}$,
  the second one uses that $y_{ij}$ is a connecting morphism, and the
  last equation uses the definition of $(\dash)^{\dag,i}$. \qed
\end{enumerate}
\doendproof
\begin{notation}
  \begin{enumerate}
  \item By Lemma \ref{lem:compfamalg_to_elgot}, for every compatible
    family \eqref{eq:compfamalg} of ffg-Elgot algebras, the morphisms
    $h_i\colon Y_i\to A$ form a cocone and thus induce a unique morphism
    $h\colon Y\to A$ with $h_i = h\o y_i$ for all $i\in I$.
  \item For every ffg equation $e\colon X \to FX + Y + A$ there exists a
    factorization
    \[ e = (\,X\xra{e_i} FX+Y_i+A\xra{FX+y_i+A} FX+Y+A\,) \]
    with $i\in I$. We put $e^\ddagger := e_i^{\dag,i}$ (and prove below
    that this is independent of the choice of $i$). 
  \end{enumerate}
\end{notation}

\begin{lemma}\label{lem:compfamalg_to_elgotfy}
  Every compatible family \eqref{eq:compfamalg} of ffg-Elgot algebras
  induces an ffg-Elgot algebra $(A,[a,h],\ddagger)$.
\end{lemma}
\begin{proof}
  We first observe that the factorization of $e$ exists because
  $(FX+Y_i+A\xra{FX+y_i+A} FX+Y+A)_{i\in I}$ is a filtered colimit
  cocone and $X$, being an ffg object, is finitely presentable. Let us
  show that $\ddagger$ well-defined, i.e.~independent of the choice of
  the factorization. To see this, suppose that another factorization
  $e=(FX+y_j+A)\o e_j$ is given. Since $D_Y$ is filtered, there exists
  $k\in I$ and connecting morphisms $y_{ik}\colon Y_i\to Y_k$ and
  $y_{jk}\colon Y_j\to Y_j$ with
  $e_k := (FX+y_{ik}+A)\o e_i = (FX+y_{jk}+A)\o e_j$. Then
  compatibility of the given family of ffg-Elgot algebras shows that
\[ e_i^{\dag,i} = e_k^{\dag,k} = e_j^{\dag,j},\]
as required.

It remains to show that $(A,[a,h],\ddagger)$ is an ffg-Elgot algebra.

\medskip
\emph{Solution.} Consider the following diagram:
\[
  \xymatrix@C+2em{
    X \ar[r]^{e^\ddag} \ar[d]^{e}
    &
    A
    \\
    FX+Y+A \ar[r]^{Fe^\ddag+Y+A}
    &
    FA+Y+A \ar[u]^{[a,h,A]}
    \\
    FX+Y_i+A
    \ar[u]_{FX+y_i+A}
    \ar[r]_-{Fe^\ddag+Y_i+A}
    \ar@{<-} `l[u] `[uu]^{e_i} [uu] 
    &
    FA+Y_i+A
    \ar[u]^{FA+y_i+A}
    \ar `r[u] `[uu]_{[a,h_i,A]} [uu]
}
\]
Its outside commutes because $e^\ddag = e_i^{\dag,i}$ and
$e_i^{\dag,i}$ is a solution of $e_i$. All other parts except,
perhaps, the upper one commute trivially. Therefore, the upper part commutes,
showing that $e^\ddag$ is a solution of $e$.

\medskip \emph{Weak Functoriality.} Suppose that we are given
ffg-equations $e\colon X \to FX + Y + Z$ and
$f\colon X' \to FX' + Y + Z$, where $Z$ is an ffg object, a
coalgebra homomorphism $m$ from $e$ to $f$, and a morphism
$g: Z \to A$. We choose factorizations
\[
  \xymatrix{
    X \ar[r]^-{e} \ar[rd]_-{e_i} & FX + Y + Z \\
    & FX + Y_i + Z \ar[u]_{FX + y_i + Z}
  }
  \qquad
  \xymatrix{
    X' \ar[r]^-{f} \ar[rd]_-{f_i} & FX' + Y + Z \\
    & FX' + Y_i + Z \ar[u]_{FX' + y_i + Z}
  }
\]
for some $i \in I$; note that we may choose the same $i$ for both $e$
and $f$ since $D_Y$ is filtered. Then in the following diagram the
outside and all inner parts except the left-hand square commute:
\[
  \xymatrix{
    X \ar[r]^-{e_i} \ar[d]_m
    &
    FX + Y_i + Z \ar[rr]^-{FX + y_i + Z}
    \ar[d]^{Fm + Y_i + Z}
    &&
    FX + Y + Z
    \ar[d]^{Fm + Y + Z}
    \ar@{<-} `u[l] `[lll]_-e [lll]
    \\
    X'
    \ar[r]_-{f_i}
    &
    FX' + Y_i + Z
    \ar[rr]_-{FX' + y_i + Z}
    &&
    FX' + Y + Z
    \ar@{<-} `d[l] `[lll]^-f [lll]
    }
\]
Hence, it follows that the two morphisms
\[
  (Fm+Y_i+Z)\o e_i,\; f_i\o m\colon X\to FX' + Y_i+ Z
\]
are merged by the colimit injection $F\ol X + y_i + Z$. Since $X$ is
an ffg object and $D_Y$ is filtered, some
connecting morphism $FX+y_{ij}+Z$ with $j\in I$ merges them, too.
Put
\[
  e_j := (FX+y_{ij}+Z)\o e_i
  \quad\text{and}\quad
  f_j := (F\ol X + y_{ij}+Z)\o \ol f_i.
\]
Then the outside of the following diagram commutes:
\[
  \xymatrix{
    X \ar[r]^-{e_i} \ar[d]_m
    &
    FX + Y_i + Z \ar[rr]^-{FX + y_{ij} + Z}
    \ar[d]^{Fm + Y_i + Z}
    &&
    FX + Y_j + Z
    \ar[d]^{Fm + Y_j + Z}
    \ar@{<-} `u[l] `[lll]_-{e_j} [lll]
    \\
    X'
    \ar[r]_-{f_i}
    &
    FX' + Y_i + Z
    \ar[rr]_-{FX' + y_{ij} + Z}
    &&
    FX' + Y_j + Z
    \ar@{<-} `d[l] `[lll]^-{f_j} [lll]
    }
\]
Now observe that $g \aft e$ factorizes through $g \aft e_i$ as follows:
\[
  \xymatrix{
    X \ar[r]^-e \ar[rd]_-{e_i}
    &
    FX + Y + Z
    \ar[rr]^-{FX + Y + g}
    &&
    FX + Y + A
    \ar@{<-} `u[l] `[lll]_-{g \aft e} [lll]
    \\
    &
    FX + Y_i + Z
    \ar[rr]_-{FX + Y_i + g}
    \ar[u]_{FX + y_i + Z}
    &&
    FX + Y_i + A
    \ar[u]_{FX + y_i + A}
    \ar@{<-} `d[l] `[lll]^-{g \aft e_i} [lllu]
    }
\]
Similarly for $g \aft f$. Furthermore note that $(FX + y_{ij} + A) \o
(g \aft e_i) = g \aft e_j$:
\[
  \xymatrix{
    X
    \ar[r]^-{e_i}
    \ar[rd]_-{e_j}
    &
    FX + Y_i + Z
    \ar[d]^{FX + y_{ij} + Z}
    \ar[rr]^-{FX + Y_i + g}
    &&
    FX + Y_i + A
    \ar[d]^{FX + y_{ij} + A}
    \ar@{<-} `u[l] `[lll]_-{g \aft e_i} [lll]
    \\
    &
    FX + Y_j + Z
    \ar[rr]_-{FX + Y_j + g}
    &&
    FX + Y_j + A
    \ar@{<-} `d[l] `[lll]^-{g \aft e_j} [lllu]
    }
\]
and similarly $(FX + y_{ij} + A) \o (g\aft f_i) = g\aft f_j$.
Thus, we obtain the Weak Functoriality of $\ddag$ from that of $\dag,i$:
\begin{align*}
  (g \aft e)^\ddag &= (g \aft e_i)^{\dag,i} & \text{def.~of $\ddag$} \\
  &= ((FX + y_{ij} + A) \o (g \aft e_i))^{\dag,j} & \text{compatibility} \\
  &= (g \aft e_j)^{\dag,j} \\
  &= (g \aft f_j)^{\dag,j} \o m & \text{Weak Functoriality of
    $\dag,j$} \\
  &= ((FX + y_{ij} + A) \o (g \aft f_i))^{\dag,j} \o m \\  
  &= (g \aft f_i)^{\dag,i}\o m & \text{compatibility} \\
  &= (g \aft f)^\ddag & \text{def.~of $\ddag$} 
\end{align*}

\medskip
\emph{Compositionality.} Let $e\colon X\to FX+Y+A$ and $f\colon Z\to FZ+Y+A$ be two ffg-equations. Factorize $e = (FX+y_i+A)\o e_i$ and $f = (FX+y_i+A)\o f_i$ with $i\in I$. Then
\begin{align*}
(f\sq e)^\ddag\o \inl &= (f_i\sq e_i)^{\dag,i}\o \inl\\
&= (f_i^{\dag,i}\bullet e_i)^{\dag,i}\\
&= (f^{\ddag}\bullet e_i)^{\dag,i}\\
&= (f^\ddag\bullet e)^\ddag
\end{align*}
Here the first equation uses the definition of $\ddag$ and the fact
that $f\sq e = (FX+y_i+A)\o (f_i\sq e_i)$. The second equation is
compositionality of $(\dash)^{\dag,i}$, the third one uses that
$f^\ddag = f_i^{\dag,i}$ by the definition of $\ddag$, and the last
equation uses the definition of $\ddag$ and the fact that
$(f^\ddag\bullet e) = (FX+y_i+A)\o (f^\ddag\bullet e_i)$.
\end{proof}

\paragraph{Proof of \autoref{thm:initialfy}.}
In order to complete the proof of Theorem~\ref{thm:corr} (and
therefore that of Theorem~\ref{thm:initialfy}), observe that the
constructions of Lemma~\ref{lem:elgot_to_compfamalg}
and~\ref{lem:compfamalg_to_elgot} are mutually inverse; the proof is
completely analogous to parts (3a) and (3b) of the proof in
Subsection~\ref{S:partA}. Moreover, the constructions of
Lemma~\ref{lem:compfamalg_to_elgotfy}
and~\ref{lem:elgotfy_to_compfamalg} are
clearly mutually inverse.

%
%
%

\subsection{Free FFG-Elgot Algebras}
\label{S:free}

We will now prove that for a free object $Y$ of $\C$ the free
ffg-Elgot algebra on $Y$ is given by the locally ffg fixed point
$\phi(F(-) + Y)$. We begin with a consequence of
Theorem~\ref{thm:initialfy}. For the forgetful functor of ffg-Elgot
algebras
\[ U_F\colon \ffgElgot F\to \C \] recall that the category
$Y\downarrow U_F$ has as objects all morphisms
$y\colon Y\to U_F(A,a,\dag)$, and morphisms into
$y'\colon Y\to U_F(B,b,\ddag)$ are the solution-preserving morphisms
$p\colon (A,a, \dag)\to (B,b,\ddag)$ with $p\o y = p'$. Denote by
$\pi\colon Y\downarrow U_F \to \C$ the projection functor given by
$\pi(y) = A$.
\begin{proposition}\label{prop:catiso}
  For every free object $Y$ of $\C$ there is an isomorphism $I$ of
  categories
  \iffull
  making the following triangle commutative:
  \[
    \xymatrix@-1pc{
      \ffgElgot (F(-)+Y) \ar[dr]_-{U_{F(-)+Y}} \ar[rr]^-I 
      && 
      Y\downarrow U_F \ar[dl]^-\pi
      \\
      & \C
    }
  \]
  \else
  such that
  \[
    U_{F(-) + Y} = (\xymatrix@1{
      \ffgElgot (F(-)+Y)
      \ar[r]^-I & Y\downarrow U_F \ar[r]^-\pi & \C
    }).
  \]
  \fi
  It is given by $(A,[a,h],\ddag) \;\mapsto\; (h\colon Y\to U_F(A,a,\dag))$.
\end{proposition}
\begin{proof}
  Using \autoref{thm:initialfy}, we just need to verify for every 
  pair of ffg-Elgot algebras $(A,[a,h],\ddag)$ and
  $(A',[a',h'],\ddag')$ that a morphism $p\colon A\to A'$ is
  solution-preserving for $F(\dash)+Y$ iff it is solution-preserving
  for $F$ and satisfies $h'=p\o h$.

  \medskip\noindent
  ($\To$) If $p$ is solution-preserving for $F(\dash)+Y$,
  then by Lemma \ref{lem:solpres} it is a homomorphism, i.e.
  $p\o [a,h] = [a',h']\o Fp$. This implies $p\o h = h'$. Moreover, for
  every ffg-equation $e\colon X\to FX+A$ the ffg-equation
  \[
    \ol e = X \xrightarrow{e} FX + A \xrightarrow{[\inl,\inr]} FX + Y
    + A
  \]
  satisfies
  $p\o \ol e^\ddag = (p\bullet \ol e)^\ddag = \ol{p\bullet e}^\ddag$,
  using~\eqref{eq:barprop}, that is, $p\o e^\dag = (p\bullet e)^\dag$.

  \medskip\noindent
  ($\Leftarrow$) If $p$ is solution-preserving for $F$ and
  $h'=p\o h$, then for every ffg-equation $e\colon X\to FX+Y+A$ we
  know that $p\o e_h^\dag = (p\bullet e_h)^\dag$ (recalling $e_h$ from
  \autoref{rem:passages}(1)). In order to derive $p\o e^\ddag =
  (p\bullet e)^\ddag$, it remains to verify
  that $p\bullet e_h = (p\bullet e)_{h'}$, which follows from the following
  commutative diagram: 
  \[
    \xymatrix@C+2em{
      X \ar[r]^-e \ar[dr]_-e
      &
      FX+Y+A \ar[r]^{FX+[h,A]} \ar@{=}[d]
      &
      FX+A \ar[r]^{FX+p}
      &
      FX+A'
      \ar@{<-} `u[l] `[lll]_-{p \aft e_h} [lll]
      \\
      &
      FX+Y+A \ar[rr]_{FX+Y+p}
      &&
      FX+Y+A' \ar[u]_-{\labelstyle FX+[h',A]}
      \ar@{<-} `d[l] `[lll]^-{p \aft e} [lllu]
    }
    \vspace*{-16pt}
\]
\end{proof}

\begin{construction}\label{C:free}
  Given an object $Y$ of $\C$, we denote by $\Phi Y$ the colimit of all
  ffg-coalgebras for $F(-)+Y$, that is, $\Phi Y = \phi(F(-)+Y)$.  Its
  coalgebra structure is invertible~\cite{Urbat17}, and we denote by \iffull
\[ t_Y\colon F\Phi Y \to \Phi Y\qquad\text{and}\qquad \eta_Y\colon
  Y\to \Phi Y \]
\else
$t_Y\colon F\Phi Y \to \Phi Y$ and $\eta_Y\colon Y\to \Phi Y$
\fi
the components of its inverse. 

The $F$-algebra $(\Phi Y, t_Y)$ is endowed with a canonical solution
operation $\dag$ defined as follows. Given an ffg-equation
$e\colon X\to FX+\Phi Y$, put
\[
  \ol e = (X\xra{e} FX+\Phi Y \xra{FX+\inl} FX+Y+\Phi Y).
\]
This ffg-equation for $F(-)+Y$ has a solution $\ol e^\ddag$ in the
ffg-Elgot algebra $\Phi Y$, and we put
\[
  e^{\dag} := (X \xrightarrow{\ol e^\ddag} \Phi Y). 
\]
\end{construction}
\begin{theorem}\label{thm:free}
 For every free object $Y$ of $\C$, the algebra $(\Phi Y, t_Y)$ with the solution operation $\dagger$ is
  a free ffg-Elgot agebra for $F$ on $Y$. 
\end{theorem}
\begin{proof}
  We prove that $\eta_Y\colon Y \to \Phi Y$ in \autoref{C:free} is the
  universal morphism. $\Phi Y$ is an ffg-Elgot algebra since, together
  with $\eta_Y$, it corresponds to the initial ffg-Elgot algebra
  $\phi(F(-)+Y)$ under the isomorphism of \autoref{prop:catiso}. This
  follows from \autoref{thm:ini} applied to $F(-) + Y$.  To verify its
  universal property, let $(A,a,\dag)$ be an ffg-Elgot algebra for $F$
  and $h\colon Y\to A$ a morphism. \autoref{prop:catiso} gives an
  ffg-Elgot algebra $(A,[a,h],\oplus)$ for $F(-)+Y$ with
  $e^\dag = \ol e^\oplus$ for all ffg-equations $e\colon X\to FX+A$
  (cf.~\autoref{rem:passages}). Furthermore, \autoref{prop:catiso}
  states that a morphism $p\colon \Phi Y\to A$ in $\C$ is
  solution-preserving w.r.t.~$F(-)+Y$ if and only if it is
  solution-preserving w.r.t.~$F$ and satisfies $p\o \eta_Y =
  h$. Therefore, the universal property of $\eta_Y\colon Y \to \Phi Y$
  w.r.t. $F$ follows from the initiality of $\Phi Y$ w.r.t.~$F(-)+Y$.
\end{proof}
\subsection{Monadicity of FFG-Elgot Algebras}
\label{S:mon}

We will now prove that the forgetful functor
$U_F\colon \ffgElgot F \to \C$ is monadic. This means that all
ffg-Elgot algebras form an algebraic category over the given variety
$\C$. To this end we must first establish that its forgetful functor
has a left-adjoint, which assigns to every object $Y$ of $\C$ a free
ffg-Elgot algebra on $Y$. So far we have seen in Theorem~\ref{thm:free}
that on every free object $Y$ we have a free ffg-Elgot algebra on
$Y$. To extend this to arbitrary objects of $\C$ we will make use of
the following result.

\begin{proposition}\label{prop:siftedcolimits}
  The forgetful functor $U_F\colon \ffgElgot F\to \C$ creates sifted
  colimits.
\end{proposition}
\proof
  Let $D\colon \D\to \ffgElgot F$ be a sifted diagram with objects
  $(A_d,a_d,(\dash)^{\dag,d})$ for $d\in \D$. Let
  \[ i_d\colon A_d\to A\quad (d\in\D) \] be a colimit cocone of
  $U_F\o D$ in $\C$. Since $F$ preserves sifted colimits, the
  forgetful functor from $\alg F$ to $\C$ creates them,
  i.e.~there exists a unique $F$-algebra structure $a\colon FA\to A$
  making every $i_d$ an $F$-algebra homomorphism:
  \[
    \xymatrix{
      FA_d  \ar[r]^{a_d} \ar[d]_{Fi_d} & A_d \ar[d]^{i_d} \\
      FA \ar[r]_a & A
    } 
  \]
  Moreover $(A,a)=\colim_{d\in \D}(A_d,a_d)$ in $\alg F$. We need
  to show that there is a unique solution operation $\dag$ on $(A,a)$
  such that $(A,a,\dag)$ is an ffg-Elgot algebra and every $i_d$ is
  solution-preserving, and moreover $i_d$ ($d\in\D$) is a colimit
  cocone in $\ffgElgot F$.

  \begin{enumerate}
  \item {Uniqueness of $\dag$.} Given a solution operation $\dag$ on
    the algebra $(A,a)$ for which all $i_d$'s are solution-preserving,
    then for every ffg-equation $e\colon X\to FX+A$, an explicit formula for
    $e^\dag$ is given as follows: since $FX+A$ is a sifted colimit of
    $FX+A_d$ ($d\in \D$) and $X$ is an ffg object, there exists a factorization
    \[  
      \xymatrix{
        X \ar[r]^-e \ar[dr]_{e_0} & FX+A \\
        & FX + A_d \ar[u]_{FX+i_d} 
      }
    \] 
    Thus $e=i_d\bullet e_0$, which implies
    \begin{equation}\label{eq:solcolimit}
      e^\dag = i_d\o e_0^{\dag,d}
    \end{equation}
    because $i_d$ is solution-preserving. This shows that $\dag$ is uniquely determined.

  \item {Existence of $\dag$.} The formula \eqref{eq:solcolimit}
    defines a solution operation $\dag$; the independence of the
    choice of the factorization is established as in the proof of
    Lemma \ref{lem:solution}. Let us verify that $(A,a,\dag)$ is an
    ffg-Elgot algebra.

    \medskip\noindent
    \emph{Solution.} $e^\dag$ is a solution of $e$:
    \[
      \xymatrix@C+2em{
        X
        \ar `u[r] `[rr]^-{e^\dag} [rr]
        \ar[r]^{e_0^{\dag,d}}
        \ar[d]^{e_0}
        &
        A_d \ar[r]^{i_d}
        &
        A
        \\
        FX+A_d
        \ar[r]^-{Fe_0^{\dag,d}+A_d} \ar[d]^{FX+i_d}
        &
        FA_d + A_d \ar[u]_{[a_d,A_d]} \ar[dr]^{Fi_d+i_d}
        \\
        FX+A \ar[rr]_-{Fe^\dag+A}
        \ar@{<-} `l[u] `[uu]^e [uu]
        & &
        FA+A \ar[uu]_{[a,A]} 
      }
    \]
    
    \medskip\noindent
    \emph{Weak Functoriality.} Suppose that we are given a coalgebra homomorphism
    \[
      \xymatrix{
        X \ar[r]^-e \ar[d]_m & FX+Z \ar[d]^{Fm+z}
        \\
        X' \ar[r]_-f & FX'+ Z
      }
    \]
    together with a morphism $h\colon Z\to A$, where $X$, $X'$ and $Z$
    are ffg objects. Factorize $h$ as in the triangle below:
    \[
      \xymatrix{
        Z \ar[r]^h \ar[dr]_{h'}& A\\
        & A_d \ar[u]_{i_d}
      }
    \]
    for some $d\in \D$.  Then the desired equality 
    \[ (h\bullet f)^\dag \o m = (h\bullet e)^\dag \]
    is established as follows:
    \begin{align*}
      (h\bullet f)^\dag\o m &= ((i_d\o h') \bullet f)^\dag \o m & \\
      &= (i_d\bullet (h'\bullet f))^\dag \o m & \text{Remark \ref{rem:square_bullet_properties}(1)} \\
      &= i_d \o (h'\bullet f)^{\dag,d}\o m & \text{def. $\dag$} \\
      &= i_d \o (h'\bullet e)^{\dag,d} & \text{$(\dash)^{\dag,d}$ weakly funct.} \\
      &= (i_d\bullet (h'\bullet e))^\dag & \text{def. $\dag$} \\
      &= ((i_d\o h')\bullet e)^\dag & \text{Remark \ref{rem:square_bullet_properties}(1)} \\
      &= (h\bullet e)^\dag &
    \end{align*}

    \medskip\noindent
    \emph{Compositionality.} Given ffg-equations
    $e\colon X\to FX+Y$ and $f\colon Y\to FY+A$, 
    factorize $f$ as follows:
    \[
      \xymatrix{
        Y \ar[r]^-f \ar[dr]_{f_0}& FY+A\\
        & FY+A_d \ar[u]_{FY+i_d}
      }
    \]
    for some $d\in \D$. Then we obtain
    \begin{align*}
      (f^\dag\bullet e)^\dag &= ((i_d\o f_0^{\dag,d})\bullet e)^\dag & \text{def. $\dag$} \\
      &= (i_d\bullet (f_0^{\dag,d} \bullet e))^\dag & \text{Remark \ref{rem:square_bullet_properties}(1)} \\
      &= i_d \o (f_0^{\dag,d}\bullet e)^{\dag,d} & \text{def. $\dag$} \\
      &= i_d \o (e\sq f_0)^{\dag,d}\o \inl & \text{$(\dash)^{\dag,d}$ compositional} \\
      &= (i_d\bullet (e\sq f_0))^\dag \o \inl & \text{def. $\dag$} \\
      &= (e\sq\, (i_d\bullet f_0))^\dag \o \inl & \text{Remark \ref{rem:square_bullet_properties}(2)} \\
      &= (e\sq f)^\dag \o \inl & 
    \end{align*}
    This completes the proof that $(A,a,\dag)$ is an ffg-Elgot algebra. 

  \item We prove that $(A,a,\dag)$ is a colimit of
    $(A_d,a_d,(\dash)^{\dag,d})$ ($d\in \D$). Thus suppose that an
    ffg-Elgot algebra $(B,b,\ddag)$ and a cocone of
    solution-preserving morphisms $m_d\colon A_d\to B$ ($d\in \D$) are
    given. We need to show that the unique morphism $m\colon A\to B$
    with $m\o i_d = m_d$ for all $d$ is solution-preserving. To this
    end, suppose that $e\colon X\to FX+A$ is an ffg-equation,
    factorized as follows:
    \[
      \xymatrix{
        X \ar[r]^-e \ar[dr]_{e_0}& FX+A\\
        & FX+A_d \ar[u]_{FX+i_d}
      }
    \]
    Then we obtain
    \begin{align*}
      (m\bullet e)^\ddag &= (m\bullet (i_d\bullet e_0))^\ddag & \\
      &= ((m\o i_d)\bullet e_0)^\ddag & \text{Remark \ref{rem:square_bullet_properties}(1)} \\
      &= (m_d\bullet e_0)^\ddag & \text{since $m\o i_d = m_d$} \\
      &= m_d \o e_0^{\dag,d} & \text{$m_d$ solution-preserving} \\
      &= m \o i_d \o e_0^{\dag,d} & \text{since $m\o i_d = m_d$} \\
      &= m\o e^\dag & \text{def. $\dag$}
    \end{align*}
    This completes the proof.\qed
  \end{enumerate}
\doendproof

\begin{theorem}\label{thm:mon}
  The forgetful functor $U_F\colon \ffgElgot F \to \C$ is monadic. 
\end{theorem}
\proof
\begin{enumerate}
\item $U_F$ has a left adjoint. Indeed, for every ffg object $Y$ we
  have a free ffg-Elgot algebra $\Phi Y$ by \autoref{thm:free},
  which defines the corresponding functor
  \[ \Phi\colon \C_\ffg \to \ffgElgot F. \] We can extend it to a left
  adjoint of $U_F$ as follows. Given an object $Y$ of $\C$, express
  it as a sifted colimit $y_i\colon Y_i\to Y$ ($i\in I$) of ffg
  objects (see Section~\ref{sec:vars}). The image of that sifted
  diagram under $\Phi$ has a colimit $\colim_{i\in I} \Phi Y_i$ in the
  category $\ffgElgot F$ by Proposition \ref{prop:siftedcolimits}. It
  follows immediately that this colimit is a free ffg-Elgot algebra on
  $Y$.

\item By Beck's Theorem (see, e.g.~\cite[Theorem~4.4.4]{Borceux2}) it
  remains to prove that $U_F$ creates coequalizers of $U_F$-split
  pairs of morphisms. These are pairs
  $f,g\colon (A,a,\dag)\to (B,b,\ddag)$ of morphisms of ffg-Elgot
  algebras such that morphisms $c\colon B\to C$, $s\colon C\to B$ and
  $t\colon B\to A$ in $\C$ are given with $c\o f = c\o g$,
  $c\o s = \id_C$, $g\o t = \id_B$ and $s\o c = f\o t$.
  \[
    \xymatrix{
      A \ar@<3pt>[r]^-f \ar@<-3pt>[r]_-g
      &
      B \ar[r]^-c
      \ar@/^1.5pc/[l]^-t
      &
      C \ar@/^1pc/[l]^-s
    }
  \]
  Since $F$ is a finitary functor, the forgetful functor from $\alg F$
  to $\C$ is monadic, see \cite{barr}. Thus, by Beck's Theorem, there
  is a unique structure $\gamma\colon FC\to C$ such that $c$ is an
  $F$-algebra homomorphism from $(B,b)$ to $(C,\gamma)$; moreover, $c$
  is a coequalizer of $f$ and $g$ in $\alg F$. We need to show that
  there is a unique solution operator $\ast$ for the algebra
  $(C,\gamma)$ such that $(C,\gamma,\ast)$ is an ffg-Elgot algebra and
  $c$ is solution-preserving, and that $c$ is then a coequalizer of
  $f$ and $g$ in $\ffgElgot F$.

  Given an ffg-equation $e\colon X\to FX+C$, we define 
  \[
    e^\ast = (X \xrightarrow{(s\bullet e)^\ddag} B \xrightarrow{c} C).
  \]
  Then $c$ is solution-preserving:
  \begin{align*}
    (c\bullet e)^\ast &= c\o (s\bullet (c\bullet e))^\ddag & \text{def. $\ast$} \\
    &= c\o ((s\o c)\bullet e)^\ddag & \text{Remark \ref{rem:square_bullet_properties}(1)} \\
    &= c\o ((f\o t)\bullet e)^\ddag & \text{$s\o c = f\o t$} \\
    &= c \o (f\bullet(t\bullet e))^\ddag & \text{Remark \ref{rem:square_bullet_properties}(1)} \\
    &= c\o f \o (t\bullet e)^\dag & \text{$f$ solution-preserving} \\
    &= c\o g \o (t\bullet e)^\dag & \text{$c\o f = c\o g$} \\
    &= c\o (g\bullet (t\bullet e))^\ddag & \text{$g$ solution-preserving} \\
    &= c\o ((g\o t)\bullet e)^\ddag & \text{Remark \ref{rem:square_bullet_properties}(1)} \\
    &= c\o e^\ddag & \text{$g\o t = \id$} 
  \end{align*}
  We prove that $\ast$ satisfies the axioms of an ffg-Elgot algebra,
  and that it is the unique ffg-Elgot algebra structure on
  $(C,\gamma)$ for which $c$ is solution-preserving.

  \medskip\noindent
  (a) $e^\ast$ is a solution of $e$:
  \[
    \xymatrix@C+2em{
      X  \ar `u[r] `[rrr]^-{e^\ast} [rrr] \ar[rr]^{(s\bullet e)^\ddag} \ar[dd]_e \ar[dr]^{s\bullet e} & & B \ar[r]^c & C \\
      & FX+B \ar[r]_{F(s\bullet e)^\ddag+B}  & FB+B \ar[u]_{[b,B]} \ar[dr]^{Fc+c} & \\
      FX+C \ar[ur]_{FX+s} \ar[rrr]_{Fe^\ast + C} & & & FC+C \ar[uu]_{[\gamma,C]} 
    }
  \]
  All inner parts of this diagram commute; for the left-hand component of
  the right-hand part, use that $c$ is solution-preserving and thus a
  homomorphism of $F$-algebras by Lemma \ref{lem:solpres}.
  
  \medskip\noindent
  (b) \emph{Weak Functoriality}. Suppose that we have a coalgebra homomorphism
  \[
    \xymatrix{
      X \ar[r]^-e \ar[d]_m & FX+Z \ar[d]^{Fm+Z} \\
      Y \ar[r]_-{f} & FY+Z 
    }
  \]
  and a morphism $h\colon Z\to C$ where $X$, $Y$ and $Z$ are ffg objects. Then
  \begin{align*}
    (h\bullet e)^\ast &= c\o (s\bullet (h\bullet e))^\ddag & \text{def. $\ast$} \\
    &= c\o ((s\o h)\bullet e)^\ddag & \text{Remark \ref{rem:square_bullet_properties}(1)} \\
    &= c\o ((s\o h)\bullet f)^\ddag \o m & \text{$\ddag$ weakly functorial} \\
    &= c\o (s\bullet (h\bullet f))^\ddag \o m & \text{Remark \ref{rem:square_bullet_properties}(1)} \\
    &= (h\bullet f)^\ast\o m & \text{def. $\ast$}
  \end{align*}

  \medskip\noindent
  (c) \emph{Compositionality}. Given ffg-equations
  $e\colon X\to FX+Y$ and $f\colon Y\to FY+C$ we compute
  \begin{align*}
    (f^\ast \bullet e)^\ast &= ((c\o (s\bullet f)^\ddag)\bullet e)^\ast & \text{def. $\ast$} \\
    &= (c\bullet ((s\bullet f)^\ddag \bullet e))^\ast & \text{Remark \ref{rem:square_bullet_properties}(1)} \\
    &= c\o ((s\bullet f)^\ddag\bullet e)^\ddag & \text{$c$ solution-preserving} \\
    &= c\o (e \sq\, (s\bullet f))^\ddag \o \inl & \text{$\ddag$ compositional} \\
    &= c\o (s\bullet (e\sq f))^\ddag \o \inl & \text{Remark \ref{rem:square_bullet_properties}(2)} \\
    &= (e\sq f)^\ast \o \inl & \text{def. $\ast$} 
  \end{align*}
  
  \medskip\noindent
  (d)~We show the uniqueness of $\ast$. Suppose that
  $+$ is another solution operation for $(C,\gamma)$ such that $c$ is
  solution-preserving.  Then
  \begin{align*}
    e^\ast &= c\o (s\bullet e)^\ddag & \text{def. $\ast$} \\
    &= (c\bullet (s\bullet e))^+ & \text{$c$ solution-preserving} \\
    &= ((c\o s)\bullet e)^+ & \text{Remark \ref{rem:square_bullet_properties}(1)} \\
    &= e^+ & \text{$c\o s = \id$}
  \end{align*}
  
  \medskip\noindent (e)~We finally show that $c$ is a coequalizer of $f$ and
  $g$. Let $m\colon (B,b,\ddag)\to (D,d,+)$ be a
  solution-preserving morphism with $m\o f = m\o g$. Since $\C$ is an
  (absolute) coequalizer in $\C$, there exists a unique morphism
  $h\colon C\to D$ with $h\o c = m$. We only need to show that it is
  solution-preserving. Indeed, given an ffg-equation $e\colon X\to
  FX+C$, we compute:
  \begin{align*}
    h\o e^\ast &= h\o c\o (s\bullet e)^\ddag & \text{def. $\ast$} \\
    &= m\o (s\bullet e)^\ddag & \text{$h\o c = m$} \\
    &= (m\bullet (s\bullet e))^+ & \text{$m$ solution-preserving} \\
    &= ((m\o s)\bullet e)^+ & \text{Remark \ref{rem:square_bullet_properties}(1)} \\
    &= ((h\o c\o s)\bullet e)^+ & \text{$h\o c=m$} \\
    &= (h\o e)^+ & \text{$c\o s=\id$} \tag*{\qed}
  \end{align*}
\end{enumerate}
\doendproof

\begin{corollary}
  The forgetful functor $W_F\colon \ffgElgot F \to \alg F$ is monadic.
\end{corollary}
\begin{proof}
  Indeed, we have a commutative triangle
  \[
    \xymatrix@C-1.5pc@R-.5pc{
      \ffgElgot F \ar[rr]^-{W_F} \ar[rd]_-{U_F}&& \alg F
      \ar[ld]^-{V_F} \\
      & \C
    }
  \]
  of forgetful functors, where $U_F$ and $V_F$ are monadic. By
  \autoref{prop:siftedcolimits} we know that $\ffgElgot F$ has
  reflexive coequalizers. Thus by~\cite[Corollary~4.5.7 and
  Exercise~4.8.6]{Borceux2}, $W_F$ is monadic, too.
\end{proof}

\section{Conclusions and Further Work}
\label{S:con}

For a functor $F$ on a variety $\C$ preserving sifted colimits, the concept
of an Elgot algebra~\cite{amv_elgot} has a natural weakening obtained by
working with iterative equations having ffg objects of variables. We
call such algebras ffg-Elgot algebras. We have proved that the locally
ffg fixed point $\phi F$, constructed by taking the
colimit of all $F$-coalgebras with an ffg carrier, is the initial
ffg-Elgot algebra for $F$. Furthermore, all free
ffg-Elgot algebras exist, and the colimit of all
ffg-coalgebras for $F(-) + Y$ yields a free ffg-Elgot algebra
on $Y$, whenever $Y$ is a free object of $\C$ on some (possibly
infinite) set. Finally, we have proved that the forgetful functor
from the category of ffg-Elgot algebras to $\C$ is monadic.

An open problem is giving a coalgebraic construction of free ffg-Elgot
algebras over arbitrary objects $Y$, similarly to \autoref{C:free},
which only works for \emph{free} object $Y$,
cf.~\autoref{thm:free}. In addition, the study of the properties of
the ensuing free ffg-Elgot algebra monad is also left for the
future. The monad of ordinary free Elgot algebras
(cf.~Section~\ref{sec:elgot}) was proved~\cite{amv_elgot} to be the
free Elgot monad on the given endofunctor $F$. It would be interesting
to see whether the above monad of free ffg-Elgot algebras is
characterized by a similar universal property.

Finally, in the current setting we have the following forgetful functors:
\[
  \ffgElgot F \to \alg F \to \C \to \Set.
\]
Each of those functors has a left-adjoint and is in fact
monadic, and we have shown that the composite of the first two is
monadic, too. We leave the question whether the
composite of all three functors is monadic for further work.

%
%
\bibliographystyle{splncs03}
\bibliography{refs}

\begin{thebibliography}{10}
\providecommand{\url}[1]{\texttt{#1}}
\providecommand{\urlprefix}{URL }

\bibitem{aamv}
Aczel, P., Ad\'amek, J., Milius, S., Velebil, J.: Infinite trees and completely
  iterative theories: A coalgebraic view. Theoret.~Comput.~Sci.  300,  1--45
  (2003), {\bfseries Fundamental study}

\bibitem{aav}
Aczel, P., Ad\'amek, J., Velebil, J.: A coalgebraic view of infinite trees and
  iteration. In: Proc.~Coalgebraic Methods in Computer Science (CMCS'01).
  Electron.~Notes Theor.~Comput.~Sci., vol.~44, pp. 1--26 (2001)

\bibitem{amm18}
Ad\'amek, J., Milius, S., Moss, L.S.: Fixed points of functors.
  J.~Log.~Algebr.~Methods Program.  95,  41--81 (2018),
  \\\mbox{\url{https://doi.org/10.1016/j.jlamp.2017.11.003}}

\bibitem{amsw19_1}
Ad\'amek, J., Milius, S., Sousa, L., Wi\ss\/mann, T.: On finitary functors.
  Theory Appl.~Categ.  34,  1134--1164 (2019)

\bibitem{amu18_cmcs}
Ad\'amek, J., Milius, S., Urbat, H.: On algebras with effectful iteration. In:
  C{\^i}rstea, C. (ed.) Proc.~Coalgebraic Methods in Computer Science (CMCS).
  Lecture Notes Comput.~Sci., vol. 11202, pp. 144--166. Springer (2018)

\bibitem{amv_elgot}
Ad\'amek, J., Milius, S., Velebil, J.: Elgot algebras.
  Log.~Methods~Comput.~Sci.  2(5:4),  31 pp. (2006)

\bibitem{amv_atwork}
Ad\'amek, J., Milius, S., Velebil, J.: Iterative algebras at work.
  Math.~Structures Comput.~Sci.  16(6),  1085--1131 (2006)

\bibitem{amv_horps_full}
Ad\'amek, J., Milius, S., Velebil, J.: Semantics of higher-order recursion
  schemes. Log.~Methods~Comput.~Sci.  7(1:15),  43 pp. (2011)

\bibitem{ar}
Ad\'{a}mek, J., Rosick\'y, J.: Locally presentable and accessible categories.
  Cambridge University Press (1994)

\bibitem{AdamekEA10}
Ad\'{a}mek, J., Rosick\'y, J., Vitale, E.: What are sifted colimits? Theory
  Appl.~Categ.  23,  251--260 (2010)

\bibitem{arv}
Ad\'{a}mek, J., Rosick\'y, J., Vitale, E.: Algebraic Theories. Cambridge
  University Press (2011)

\bibitem{applegate}
Applegate, H.: Acyclic models and resolvent functors. Ph.D. thesis, Columbia
  University (1965)

\bibitem{barr}
Barr, M.: Coequalizers and free triples. Math.~Z.  116,  307--322 (1970)

\bibitem{br88}
Berstel, J., Reutenauer, C.: Rational Series and Their Languages.
  Springer-Verlag (1988)

\bibitem{BE93}
Bloom, S.L., \'Esik, Z.: Iteration Theories: the equational logic of iterative
  processes. EATCS Monographs on Theoretical Computer Science, Springer (1993)

\bibitem{bms13}
Bonsangue, M.M., Milius, S., Silva, A.: Sound and complete axiomatizations of
  coalgebraic language equivalence. ACM Trans.~Comput.~Log.  14(1:7),  52 pp.
  (2013)

\bibitem{Borceux2}
Borceux, F.: Handbook of Categorical Algebra, vol.~2. Cambridge University
  Press (1994)

\bibitem{courcelle}
Courcelle, B.: Fundamental properties of infinite trees. Theoret.~Comput.~Sci.
  25,  95--169 (1983)

\bibitem{DaveyD85}
Davey, B.A., Davis, G.: Tensor products and entropic varieties. Algebra
  Universalis  21,  68--88 (1985)

\bibitem{Doberkat06}
Doberkat, E.: Eilenberg-moore algebras for stochastic relations. Inf. Comput.
  204(12),  1756--1781 (2006), erratum and addendum published in {\em
  Inf.~Comput.} 206(12), 1476--1484 (2008)

\bibitem{DrosteEA09}
Droste, M., Kuich, W., Vogler, H. (eds.): Handbook of weighted automata.
  Monographs in Theoretical Computer Science, Springer (2009)

\bibitem{Elgot75}
Elgot, C.C.: Monadic computation and iterative algebraic theories. In: Rose,
  H.E., Sheperdson, J.C. (eds.) Logic Colloquium '73. vol.~80, pp. 175--230.
  North-Holland Publishers, Amsterdam (1975)

\bibitem{em_2010}
\'Esik, Z., Maletti, A.: Simulation vs.\ equivalence. In: Proc.\ 6th Int.\
  Conf.\ Foundations of Computer Science. pp. 119--122. CSREA Press (2010)

\bibitem{em_2011}
\'Esik, Z., Maletti, A.: Simulations of weighted tree automata. In:
  Proc.~CIAA'11. Lecture Notes Comput.~Sci., vol. 6482, pp. 321--330. Springer
  (2011)

\bibitem{fpt}
Fiore, M., Plotkin, G.D., Turi, D.: Abstract syntax and variable binding. In:
  Proc.~LICS'99. pp. 193--202. IEEE Press (1999)

\bibitem{Fliess1974}
Fliess, M.: Sur divers produits de s\'eries formelles. Bulletin de la
  Soci\'et\'e Math\'ematique de France  102,  181--191 (1974)

\bibitem{freyd68}
Freyd, P.: R\'edei's finiteness theorem for commutative semigroups.
  Proc.~Amer.~Math.~Soc.  19(4),  p.~1003 (1968)

\bibitem{glmp_cmcs}
Ghani, N., L\"uth, C., Marchi, F.D., Power, A.J.: Algebras, coalgebras, monads
  and comonads. In: Proc.~Coalgebraic Methods in Computer Science (CMCS'01).
  Electron.~Notes Theor.~Comput.~Sci., vol.~44, pp. 128--145 (2001)

\bibitem{glmp}
Ghani, N., L\"uth, C., Marchi, F.D., Power, A.J.: Dualizing initial algebras.
  Math.~Structures Comput.~Sci.  13(2),  349--370 (2003)

\bibitem{ginali}
Ginali, S.: Regular trees and the free iterative theory. J.~Comput.~System Sci.
   18,  228--242 (1979)

\bibitem{johnstone_lift}
Johnstone, P.T.: Adjoint lifting theorems for categories of algebras.
  Bull.~London Math.~Soc.  7,  294--297 (1975)

\bibitem{Johnstone77}
Johnstone, P.T.: Topos Theory. Academic Press, London (1977)

\bibitem{KurzEA13}
Kurz, A., Petrisan, D., Severi, P., de~Vries, F.J.: Nominal coalgebraic data
  types with applications to lambda calculus. Log.~Methods Comput.~Sci.
  9(4:20),  51 pp. (2013)

\bibitem{lambek}
Lambek, J.: A fixpoint theorem for complete categories. Math.~Z.  103,
  151--161 (1968)

\bibitem{Linton66}
Linton, F.E.J.: Autonomous equational categories. J.~Math.~Mech.  15,  637--642
  (1966)

\bibitem{m_cia}
Milius, S.: Completely iterative algebras and completely iterative monads.
  Inform.~and Comput.  196,  1--41 (2005)

\bibitem{m_linexp}
Milius, S.: A sound and complete calculus for finite stream circuits. In:
  Proc.~LICS'10. pp. 449--458. IEEE Computer Society (2010)

\bibitem{milius18}
Milius, S.: Proper functors and fixed points for finite behaviour.
  Log.~Methods.~Comput.~Sci.  14(3:22),  32 pp. (2018)

\bibitem{mpw16}
Milius, S., Pattinson, D., Wi\ss\/mann, T.: A new foundation for finitary
  corecursion: The locally finite fixpoint and its properties. In:
  Proc.~FoSSaCS'16. Lecture Notes Comput.~Sci. (ARCoSS), vol. 9634, pp.
  107--125. Springer (2016)

\bibitem{mpw20}
Milius, S., Pattinson, D., Wi\ss\/mann, T.: A new foundation for finitary
  corecursion and iterative algebras. Inform.~and Comput.  271 (2020), article
  104456

\bibitem{msw16}
Milius, S., Schr\"oder, L., Wi\ss\/mann, T.: Regular behaviours with names: On
  rational fixpoints of endofunctors on nominal sets. Appl.~Categ.~Structures
  24(5),  663--701 (2016)

\bibitem{mw15}
Milius, S., Wi\ss\/mann, T.: Finitary corecursion for the infinitary lambda
  calculus. In: Proc.~CALCO'15. LIPIcs, vol.~35, pp. 336--351. Schloss Dagstuhl
  (2015)

\bibitem{Moss01}
Moss, L.S.: Parametric corecursion. Theoret.~Comput.~Sci.  260(1--2),  139--163
  (2001)

\bibitem{nelson}
Nelson, E.: Iterative algebras. Theoret.~Comput.~Sci.  25,  67--94 (1983)

\bibitem{TuriP97}
Plotkin, G.D., Turi, D.: Towards a mathematical operational semantics. In:
  Proc.~Logic in Computer Science (LICS'97). pp. 280--291 (1997)

\bibitem{redei}
R\'edei, L.: The Theory of Finitely Generated Commutative Semigroups. Pergamon,
  Oxford-Edinburgh-New York (1965)

\bibitem{rutten_rat}
Rutten, J.J.M.M.: Rational streams coalgebraically. Log.~Methods~Comput.~Sci.
  4(3:9),  22 pp. (2008)

\bibitem{schuetzenberger}
Sch\"utzenberger, M.P.: On the definition of a family of automata. Inform.~and
  Control  4(2--3),  275--270 (1961)

\bibitem{sbbr13}
Silva, A., Bonchi, F., Bonsangue, M.M., Rutten, J.J.M.M.: Generalizing
  determinization from automata to coalgebras. Log.~Methods Comput.~Sci
  9(1:9),  27 pp. (2013)

\bibitem{SilvaS11}
Silva, A., Sokolova, A.: Sound and complete axiomatization of trace semantics
  for probabilistic systems. Electr. Notes Theor. Comput. Sci.  276,  291--311
  (2011), \url{https://doi.org/10.1016/j.entcs.2011.09.027}

\bibitem{SokolovaW15}
Sokolova, A., Woracek, H.: Congruences of convex algebras. J.~Pure
  Appl.~Algebra  219(8),  3110--3148 (2015)

\bibitem{SokolovaW18}
Sokolova, A., Woracek, H.: Proper semirings and proper convex functors. In:
  Baier, C., Lago, U.D. (eds.) Proc.~FoSSaCS 2018. Lecture Notes Comput.~Sci.,
  vol. 10803, pp. 331--347. Springer (2018)

\bibitem{tiuryn80}
Tiuryn, J.: Unique fixed points vs.~least fixed points. Theoret.~Comput.~Sci.
  12,  229--254 (1980)

\bibitem{Urbat17}
Urbat, H.: Finite behaviours and finitary corecursion. In: Proc.~CALCO'17.
  LIPIcs, vol.~72, pp. 24:1--24:15. Schloss Dagstuhl (2017)

\bibitem{jcssContextFree}
Winter, J., Bonsangue, M.M., Rutten, J.J.: Context-free coalgebras.
  J.~Comput.~System Sci.  81(5),  911 -- 939 (2015)

\end{thebibliography}

\iffull
\clearpage
\appendix
\section{Appendix}

\section*{Details on the Definition of $\phi F$ (see \autoref{R:app}\ref{R:app:2})}

Recall~\cite{arv} that an object $X$ of $\C$ whose hom-functor
$\C(X,-)$ preserves sifted colimits is called \emph{perfectly
  presentable}, and that these objects are precisely the split
quotients of ffg objects. Let $\coapp F$ denote the full subcategory
of coalgebras carried by perfectly presentable objects. We show that
$\phi F$ can be defined as the colimit of all such $F$-coalgebras, in
symbols:
\[
  \phi F = \colim(\coapp F \subto \coa F). 
\]
To this end, it suffices to prove that the inclusion functor 
\[
  I\colon \coafr F \subto \coapp F
\]
is cofinal. This means that
\begin{enumerate}
\item for every coalgebra in $\coapp F$ there is a homomorphism into
  some coalgebra in $\coafr F$, and
\item for every span $(Y,d)\xleftarrow{f} (X,c) \xra{g} (Z,e)$ in
  the category $\coapp F$ with codomains in $\coafr F$, there exists a zig-zag of
  morphisms in the slice category $(X,c)/\coafr F$ connecting $f$ and $g$.
\end{enumerate}

\paragraph{Proof of (1).} Given an $F$-coalgebra $c\colon X \to FX$
with $X$ perfectly presentable, we know that $X$ is a split quotient
of some ffg object $W$ of $\C$, i.e.~we have $e\colon W \epito X$ and
$m\colon X \monoto W$ with $e\o m = \id_X$ in $\C$. Put
\[
  w:= (W \xrightarrow{e} X \xrightarrow{c} FX \xrightarrow{Fm} FW).
\]
Then $(W,w)$ is an ffg-coalgebra such that $m\colon (X,c) \monoto (W,w)$ is a coalgebra
homomorphisms as desired: 
\[
  w \o m = Fm \o c \o e \o m = Fm \o c.
\]

\paragraph{Proof of (2).} Now suppose we have two coalgebra
homomorphisms $f\colon (X,c) \to (Y,d)$ and $g\colon (X,c) \to (Z,e)$
where $X$ is perfectly presentable and $Y$ and $Z$ are ffg objects. As in the
proof of~(1), choose $e$ and $m$ and form the ffg-coalgebra
$(W,w)$. Now observe that $e\colon (W,w)\epito (X,c)$ is a coalgebra
homomorphism:
\[
  Fe \cdot w = Fe \cdot Fm \cdot c \cdot e = c \cdot e.
\]
Due to $e\o m = \id_X$, we then have the
following zig-zag relating $f$ and $g$:
\[
  \xymatrix{
    & (X,c)
    \ar[ld]_f \ar[rd]^g \ar[d]^m \\
    (Y,d) 
    & 
    (W,w) \ar[l]_{f\o e} \ar[r]^{g\o e}
    & 
    (Z,e)
  }
\]


\fi

\end{document}
